\documentclass[reqno,12pt]{amsart}
\usepackage[margin=1in]{geometry}
\usepackage{multicol}
\usepackage{lipsum}
\usepackage{mwe}

\usepackage{graphicx}
\usepackage{amsmath}
\usepackage{amssymb}
\usepackage{amsthm}
\usepackage{enumitem}
\usepackage{bbm}
\usepackage[l2tabu,orthodox]{nag}
\usepackage[all,error]{onlyamsmath}
\usepackage[strict]{csquotes}
\usepackage{url}
\usepackage{color}

\allowdisplaybreaks

\newtheorem{theorem}{Theorem}
\newtheorem{lemma}[theorem]{Lemma}
\newtheorem{corollary}[theorem]{Corollary}

\theoremstyle{definition}
\newtheorem{definition}[theorem]{Definition}
\newtheorem{remark}[theorem]{Remark}
\newtheorem{example}[theorem]{Example}

\newcommand{\re}{\operatorname{Re}}

\newcommand{\NN}{\mathbb{N}}

\newcommand{\ZZ}{\mathbb{Z}}
\newcommand{\RR}{\mathbb{R}}
\newcommand{\CC}{\mathbb{C}}

\newcommand{\rmd}{\mathrm{d}}
\newcommand{\rmi}{\mathrm{i}}
\newcommand{\rmr}{\mathrm{right}}
\newcommand{\rml}{\mathrm{left}}

\newcommand\cB{{\mathcal{B}}}
\newcommand\cL{{\mathcal{L}}}
\newcommand\cD{{\mathcal{D}}}
\newcommand\cR{{\mathcal{R}}}

\newcommand\cH{{\mathcal{H}}}
\newcommand\cF{{\mathcal{F}}}
\newcommand\cN{{\mathcal{N}}}

\newcommand\cJ{{\mathcal{J}}}

\newcommand{\wti}{\widetilde}

\newcommand{\tv}{\wti{v}}
\newcommand{\ts}{\wti{s}}

\newcommand{\tu}{\wti{u}}

\newcommand{\tW}{\wti{W}}

\newcommand{\oor}{\mathrm{ord}}

\newcommand{\arcsinh}{\mathrm{arcsinh}}

\newcommand{\loc}{\operatorname{loc}}
\newcommand{\spec}{\operatorname{spec}}

\newcommand{\ran}{\operatorname{ran}}

\newcommand{\beq}{\begin{equation}}
	\newcommand{\eeq}{\end{equation}}

\newcommand{\im}{{\operatorname{Im}}}
\newcommand{\sgn}{{\mathrm{sign}}}

\newcommand{\hPhi}{{\widehat{\Phi}}}

\renewcommand\d{\partial}

\def\eps{\varepsilon }

\renewcommand\d{\partial}

\def\eps{\varepsilon}

\def\tmu{\tilde{\mu}}

\newcommand\br{\begin{remark}}
	\newcommand\er{\end{remark}}
\newcommand\bp{\begin{pmatrix}}
	\newcommand\ep{\end{pmatrix}}
\newcommand{\be}{\begin{equation}}
	\newcommand{\ee}{\end{equation}}
\newcommand\ba{\begin{equation}\begin{aligned}}
		\newcommand\ea{\end{aligned}\end{equation}}

\newcommand{\bap}{\begin{app}}
	\newcommand{\eap}{\end{app}}
\newcommand{\begs}{\begin{exams}}
	\newcommand{\eegs}{\end{exams}}
\newcommand{\beg}{\begin{example}}
	\newcommand{\eeg}{\end{exaplem}}
\newcommand{\bpr}{\begin{proposition}}
	\newcommand{\epr}{\end{proposition}}
\newcommand{\bt}{\begin{theorem}}
	\newcommand{\et}{\end{theorem}}
\newcommand{\bc}{\begin{corollary}}
	\newcommand{\ec}{\end{corollary}}
\newcommand{\bl}{\begin{lemma}}
	\newcommand{\el}{\end{lemma}}
\newcommand{\bd}{\begin{definition}}
	\newcommand{\ed}{\end{definition}}
\newcommand{\brs}{\begin{remarks}}
	\newcommand{\ers}{\end{remarks}}

\newcommand{\CalR}{\mathcal{R}}

\newcommand\cC{{\mathcal C}}

\newcommand{\rma}{\mathrm{a}}
\newcommand{\rmR}{\mathrm{R}}

\numberwithin{equation}{section}
\numberwithin{theorem}{section}

\begin{document}

\author{ Yuri Latushkin}
\address{University of Missouri}
\email{latushkiny@missouri.edu}
\thanks{Y.L. was supported by the NSF grant  DMS-2106157, and would like to thank the Courant Institute of Mathematical Sciences and especially Prof.\ Lai-Sang Young for the opportunity to visit CIMS.\\ A. P.   was partially supported under the Simons Foundation Grant nr. 524928.\\E-mail addresses: latushkiny@missouri.edu (Y. Latushkin), pogana@miamioh.edu (A.Pogan).}
\author{Alin Pogan}
\address{Miami University}
\email{pogana@miamioh.edu}
\title[Resonances of the square well]{Resonances for the one dimensional Schr\"odinger operator with the matrix-valued complex square-well potential}
\date{\today}
\begin{abstract}We study the resonances of (generally, non-selfadjoint) Schr\"odinger operators with matrix-valued square-well potentials. We compute explicitly the Jost function and derive complex transcendental equations for the resonances. We prove several results concerning the distribution of resonances in the complex plane. We compute the multiplicity of resonances and  prove a version of the Weyl Law for the number of resonances.
\end{abstract}	

\maketitle

\vspace{0.3cm}
\begin{minipage}[h]{0.48\textwidth}
	\begin{center}
		University of Missouri \\
		Department of Mathematics\\
		810 East Rollins Street\\ Columbia, MO 65211, USA
	\end{center}
\end{minipage}
\begin{minipage}[h]{0.48\textwidth}
	\begin{center}
		Miami University\\
		Department of Mathematics\\
		100 Bishop Circle\\
		Oxford, OH 45056, USA
	\end{center}
\end{minipage}

\vspace{0.3cm}

\section{Introduction}\label{Intro}

Schr\"odinger operators with square well potentials are classical; we were not able to trace the study of their resonances as far as to \cite{Ga} but the square well definitely takes a notable place in every text on quantum scattering
 \cite{H,Me,Sc,Si,T}
and still attract full attention of many \cite{Br,MT,MM,MP} for quite some time. The beauty of this example is that one can calculate everything by hand, and illustrate practically all notions in scattering theory in general and, in particular, in the theory of scattering resonances; we refer to \cite{DZ,Hi,Sj,Z} and also to \cite{AGHKH,C06, CD, CDY, FH, GH} for several general sources on this fascinating subject. The current paper originated as an attempt to indeed provide an illustration to the results of the upcoming work \cite{LP3}
where, in the spirit of \cite[Theorem 2.9]{DZ}, we use resonances to represent solutions of the abstract wave equations with an operator generating a cosine family (and, in particular, covering the case of non-self-adjoint Schr\"odinger operators with matrix valued potentials). We ended up writing a detailed analysis of resonances for the square well trying to make the exposition as elementary and detailed as possible with no intention to claim novelty of our results. 

We cover the case of matrix valued potentials. We are aware of only few recent papers on resonances in this case (often refer to as {\em multichannel systems}, \cite{dAR,N}), see, e.g., \cite{Mo1,Mo2} and  \cite{ABMS}, but of course there is a big general literature on the matrix Schr\"odinger operators, see, e.g., \cite{EGNT} and also  \cite{GNZ} and the bibliography therein. Specifically, 
we consider the Schr\"odinger operator $A=\d_{xx}+V(x)$ in the space
 $L^2(\RR; \CC^d)$ of $\CC^d$-valued vector function with the (matrix valued complex) square well potential 
 \begin{equation}\label{defV}
 \text{ $V(x)=a$ for $|x|\le\ell$ and $V(x)=0$ for $|x|>\ell$ where $\ell>0$ and $a\in\CC^{d\times d}$}.
 \end{equation}
Here, $a$ is a fixed (generally, non-Hermitian) $d\times d$ matrix with complex entries and $d\ge1$.
Throughout, we assume that the matrix $a$ is invertible, that is, $0\notin\spec(a)$.

In Section \ref{Jost-Green} we provide explicit formulas for the (matrix valued) Jost solutions, and also for the Jost function and the Green's function, associated with the Schr\"odinger operator $A$. Following \cite{DZ}, we call {\em resonances} all poles of the (meromorphic in the entire complex plane) Green's function; their set we denote by $\CalR(A)$. In particular, the resonances located in the right half plane correspond to the poles of the resolvent of the operator $A$ while the resonances located in the left half plane correspond to the poles of the meromorphic extension of the resolvent viewed as an operator acting from $L^2_{\textrm{comp}}(\RR; \CC^d)$ into 
$H^2_{\loc}(\RR; \CC^d)$. 
It turns out that each eigenvalue $\alpha$ of the matrix $a$ from \eqref{defV} generates infinitely many resonances of the operator $A$. This result is of course not surprising as for the general case of compactly supported real valued (scalar) potentials it goes back to \cite{R}.
The resonances generated by the eigenvalue $\alpha$ are hereby called $\alpha$-resonances; we denote their set by $\CalR_\alpha$. We show that $\CalR(A)=\cup_{\alpha\in\spec(a)}\CalR_\alpha$ and $\CalR_\alpha=\CalR(A_\alpha)$ where $A_\alpha=\partial_{xx}+\alpha\chi_{[-\ell,\ell]}$ is the Schr\"odinger operator in $L^2(\RR;\CC)$ with the {\em scalar} (complex) well potential, $\chi_{[-\ell,\ell]}$ is the characteristic function of  $[-\ell,\ell]$. We show that the set $\CalR(\alpha_1)\cap\CalR(\alpha_2)$ is finite for different eigenvalues $\alpha_1$ and $\alpha_2$.
We write down explicitly transcendental equations for the $\alpha$-resonances and give a very detailed description of the solutions to the equations.

In Section \ref{Complex-Resonances} we show that each complex eigenvalue $\alpha$ of the matrix $a$ generates infinitely many resonances (Theorem \ref{t3.1}), but that only finitely many of them are located in the complex right half-plane (Theorem \ref{t3.12}). It is interesting to compare this with the case of the real $\delta$-potential (formally obtained by letting $\RR\ni a\to+\infty$ or $-\infty$ and $\ell\to0$) where one can have exactly one real resonance \cite[p.\ 78]{AGHKH}. Moreover, we show that there are infinitely may $\alpha$-resonances located in the second and third quadrant of the complex plane (Theorem \ref{t3.14}). Furthermore, in this theorem we give a {\em precise} asymptotic formula for a sequence $\{\lambda_n\}_{n\in\NN}$ of $\alpha$-resonances located in the left half-plane.
In particular, the formula allows us to show ``by hand" that there are only finitely many $\alpha$-resonances located to the right of a logarithmic curve  (Theorem \ref{t3.14}), a standard result \cite[Theorem 2.10]{DZ} that goes back probably to \cite{LaP}. Our approach is based on a direct analysis of the (complex) transcendental equations for (complex) $\alpha$-resonances; the equations can be re-written as real systems with unknown real and imaginary parts of the complex resonances. In Appendix \ref{App} we collected some rather technical but very detailed information about the systems. 

In Section~\ref{Real-Zeros} we study in more details the resonances located on the real and imaginary axes and/or generated by the real eigenvalues of $a$, if any.  We give a very detailed information about the location of the $\alpha$-resonances of $A$ for a positive (Theorem \ref{thm: 3.1}) or negative (
Theorem \ref{t4.3}) eigenvalue $\alpha$ of the matrix $a$ depending on the value of the parameter $\ell\sqrt{|\alpha|}$.

In Section~\ref{Multiplicity} we address the multiplicity of resonances.  We begin by describing in Theorem \ref{t5.2} the resonances of $A$ (that is, zeros of the Jost function defined as the determinant of a $(d\times d)$-matrix) in terms of the eigenvalues of infinite dimensional trace class Birman-Schwinger-type operators or, equivalently, zeros of the respective infinite dimensional Fredholm determinants. This is a fairly standard result; our proof if  based on the well developed by now theory of integral operators with semi-separable kernels, see, e.g., \cite[Chapter IX]{GGK1} and an important contribution in \cite{GM}; we follow the latter.
We then define the multiplicity of a resonance using the Birman-Schwinger principle, following \cite{G,GH} rather than \cite{DZ}.
Our final formula (Theorem \ref{t5.10}) for the multiplicity of resonances of $A$ involves the algebraic multiplicities of the eigenvalues $\alpha$ of the matrix $a$ and an explicitly computed function which takes values either $1$ or $2$ depending on the location of $\alpha$. To prove the formula we have to study in great details the sets $\cR_\alpha$ of resonances generated by particular eigenvalues of $a$; this study occupies the main part of Section~\ref{Multiplicity}. 

\textbf{Notation.}  We let $\CC_+=\{\lambda\in\CC: \re\lambda>0\}$ and
$D(\lambda,\varepsilon)$ be the disc in $\CC$ centered at $\lambda$ of radius $\varepsilon$.  We denote by $\cB(X,Y)$ the set of linear bounded operators between Banach spaces $X$ and $Y$ and by $\cC\cL(X,Y)$ the set of linear continuous operators between Fr\'{e}chet spaces $X$ and $Y$,
by $\spec(K)$, $\spec_{\rm p}(K)$, $\spec_{\rm d}(K)$ and $\spec_{\rm ess}(K)$ the spectrum,
point spectrum, discrete spectrum and essential spectrum of an operator $K$, by $m_{\rm a}(\lambda,K)$ the algebraic multiplicity of an eigenvalue $\lambda$ of $K$, by ${\rm ord}(\lambda,f)$ the order of multiplicity of a zero $\lambda$ of a function $f$, by $\chi_{[-\ell,\ell]}$ the characteristic function of the segment $[-\ell,\ell]$. We adhere to ``semigroup'' notations and write $A=\partial_{xx}$ rather then $A=-\partial_{xx}$, cf.\ \cite{DZ}; consequently, the spectral parameter ${z}^{1/2}$, $\im(z^{1/2})>0$, so that $Au=zu$ from, say, \cite{DZ,GM}
and the spectral parameter $\lambda$, $\re\lambda>0$, so that $Au=\lambda^2u$ from the current paper are related via $\lambda=-\rmi z^{1/2}$. Also,
we denote by $\sqrt{z}$ the brunch of the square root obtained by cutting $\CC$ along $(-\infty,0]$ so that $\re\sqrt{z}>0$ for $z\notin(-\infty,0]$ and $\sqrt{z}=\sqrt{z+\rmi 0}$ for $z\in(-\infty,0]$.

\textbf{Acknowledgments.} We are thankful to T.\ Christiansen, F.\ Gesztesy and I.\ M.\ Karabash for useful references and suggestions that helped us to navigate in the ocean of the resonance literature.
\section{The Jost function and the Green function}\label{Jost-Green}
In this section we give formulas for the Jost solutions, Jost function, and the associated Green's function related to the matrix square well potential.
The matrix-valued eigenvalue problem reads
\begin{equation}\label{evp}
\d_{xx}U(x)=(\lambda^2-a)U(x) \text{ for $|x|\le\ell$, $\d_{xx}U(x)=\lambda^2U(x)$ for $|x|>\ell$,} \end{equation}
where $\lambda\in\CC$ and $U:\RR\to\CC^{d\times d}$. 
Our first objective is to define the Jost solutions $U_\pm(x,\lambda)$ to \eqref{evp} that exponentially decay to zero as $x\to\pm\infty$ provided $\re\lambda>0$ and are analytic for all $\lambda\in\CC$  for each $x\in\RR$.

We fix a $\lambda\in\CC$ and seek to define the matrix $b=b(\lambda)\in\CC^{d\times d}$ such that $b^2=\lambda^2I_{d\times d}-a$. We will express the Jost solutions using $b(\lambda)$, cf.\ \eqref{defUplus}.
If $d=1$, that is, $a\in\CC$, then one can define $b(\lambda)=\sqrt{\lambda^2-a}$ for all $\lambda\in\CC\setminus\Gamma$ where $\Gamma=\{\lambda\in\CC: \lambda^2-a\in(-\infty,0]\}$. 
We note that $b(\lambda)$ is not analytic (not even continuous) for $\lambda\in\Gamma$. However, we will see below that the Jost solutions are analytic in $\lambda$ because they are {\em even} functions of $b$, and thus, in fact, depend on $b^2$ rather than $b$, and therefore do not depend on the choice of the branch of the square root.

Considering the general case when $d\ge1$, that is, $a\in\CC^{d\times d}$, we temporarily assume that
\begin{equation}\label{lambdasp}
\lambda^2\notin\spec(a)\cup\{0\};
\end{equation}
this assumption will be removed soon.  We will define the matrix 
\begin{equation}\label{def-b}
b=b(\lambda) \text{ such that $b^2=\lambda^2I_{d\times d}-a$}
\end{equation}
using Riesz-Dunford functional calculus, cf., e.g., \cite[Section VII.4]{Con}.
 To begin, we use our standing assumption $0\notin\spec(a)$ and  choose a curve $\Gamma$ in $\CC$ that connects zero  and infinity such that $\Gamma\cap\spec(a)=\emptyset$.
We cut the complex plane along $\Gamma$, define the unique branch of the square root function $({\cdot})^{1/2}$ such that $({-1})^{1/2}=\rmi$ and introduce
 $\phi_\lambda(z)$ so that $\phi_\lambda(z)=({\lambda^2-z})^{1/2}$ for all $z\in\CC$. Next, we take a simple closed positively oriented contour
$\gamma$ in $\CC\setminus\Gamma$ enclosing $\spec(a)$ such that $\gamma\cap\spec(a)=\emptyset$ and $\gamma\cap\Gamma=\emptyset$ for the pre-fixed curve $\Gamma$, again using the standing assumption $0\notin\spec(a)$. The function $\phi_\lambda(\cdot)$ is analytic on $\spec(a)$ since it is analytic for $z\in\CC\setminus\Gamma$.
We now define $b=b(\lambda)$ such that $b^2=\lambda^2I_{d\times d}-a$
by the formula
\begin{equation}\label{defb}
b=b(\lambda):=\frac1{2\pi i}\int_\gamma\phi_\lambda(z)(z-a)^{-1}\,dz.
\end{equation} 

Still (temporarily) assuming \eqref{lambdasp} so that, in particular, $b=b(\lambda)$ is invertible and $\lambda\neq0$, we define the Jost solution $U_+=U_+(x,\lambda)\in\CC^{d\times d}$ to \eqref{evp}
 by letting
 \begin{equation}\label{defUplus}
 U_+(x,\lambda)= \begin{cases}
 e^{-\lambda x}I_{d\times d},& \text{ for } x>\ell, \\
  e^{-\lambda\ell}\big(\cosh((x-\ell)b)-\lambda b^{-1}\sinh((x-\ell)b)\big), & \text{ for } |x|\le\ell,\\ 
 W_+(\lambda)e^{-\lambda x}+W_-(\lambda)e^{\lambda x},  
& \text{ for } x<-\ell,
 \end{cases}\end{equation}
 where we introduce $\CC^{d\times d}$-valued functions
 \begin{equation}\label{defUplus3}
 W_+(\lambda)=e^{-2\lambda\ell}
 \big(
 \cosh(2\ell b)+\frac12(\lambda b)^{-1}(\lambda^2+b^2)\sinh(2\ell b)\big),\, W_-(\lambda)=\frac12a(\lambda b)^{-1}\sinh(2\ell b).
  \end{equation}
The last line in \eqref{defUplus} could be also re-written as follows,
 \begin{equation}\label{defUplus2}\begin{split}
 U_+(x,\lambda)=e^{-\lambda\ell}\big(\big(&\cosh(2\ell b)+\lambda b^{-1}\sinh(2\ell b)\big)\cosh(\lambda(x+\ell))\\&\quad-\big(\lambda\cosh(2\ell b)+b\sinh(2\ell b)\big)\lambda^{-1}
 \sinh(\lambda(x+\ell))\big), \text{ for } x<-\ell;
 \end{split}
 \end{equation}
 this shows  that $U_+(x,\cdot)$ has no singularity at $\lambda=0$, that is, 
 \begin{equation}
 U_+(x,0)=\begin{cases} I_{d\times d}, & \text{ for $x>\ell$,}\\
 \cosh((x-\ell)b), & \text{ for $|x|\le\ell$,}\\
 \cosh(2\ell b)-b\sinh(2\ell(x+\ell)b),& \text{ for $x<-\ell$.}
 \end{cases}
 \end{equation}

Formulas \eqref{defUplus}, \eqref{defUplus3} and \eqref{defUplus2}extend to the case when $b$ is not invertible, cf.\ Remark \ref{remb0}, with the expression $b^{-1}\sinh((x-\ell)b)$ replaced by $(x-\ell)P_0+\big(b^{-1}\sinh((x-\ell)b\big)(I-P_0)$, where $P_0=P_0(\lambda)$ is the Riesz projection for $b=b(\lambda)$ onto its zero eigenvalue, see more details in Remark \ref{remb0} below. In particular, if $d=1$ and $b(\lambda)=0$ then $b^{-1}\sinh((x-\ell)b)$, $b^{-1}\sinh(2\ell b)$, etc., can be replaced by 
 $(x-\ell)$, $2\ell$, etc.
 
 The matrix $b$ commutes with the matrices $\cosh((x-\ell)b)$ and $b^{-1}\sinh((x-\ell)b)$. The right hand side of \eqref{defUplus} is a linear combination with $x$-independent matrix coefficients of the solutions $e^{\pm\lambda x}I_{d\times d}$ and 
 $e^{\pm xb(\lambda)}$ to \eqref{evp} at the respective parts of the line $\RR$. A direct computation shows that $U_+(\cdot,\lambda)$ and $\d_xU_+(\cdot,\lambda)$ are continuous at $x=\pm\ell$. Thus, $U_+(\cdot,\lambda)$ is a solution to equation \eqref{evp}. 
 
\begin{remark}\label{remJsol}
 If $\textrm{Re }\, (\lambda)>0$ then $U_+(\cdot,\lambda)$ is asymptotic to the exponentially decaying to zero as $x\to+\infty$ ``plane wave'' $e^{-\lambda x}I_{d\times d}$. This justifies the term ``Jost solution''
 for $U_+(\cdot,\lambda)$.   Clearly, for a vector $u_0\in\CC^d$ the following equivalencies hold:  
 The $\CC^d$-valued function $U_+(\cdot,\lambda)u_0$ is an outgoing\footnote{Cf., e.g., \cite[Section 2.1]{DZ}.} solution to \eqref{evp} if and only if $u_0\in\ker (W_+(\lambda))$  if and only if  the function $U_+(\cdot,\lambda)u_0$ belongs to $L^2(\RR;\CC^d)$. 
 In particular, if $\re\lambda>0$ then the map $u_0\mapsto U_+(\cdot,\lambda)u_0$ gives an isomorphism between $\ker(W_+(\lambda))$ and $\ker(\lambda^2 I_{L^2}-A)$.
Indeed, 
 the term $e^{-\lambda x}$ in the last line of formula \eqref{defUplus} grows as $x\to-\infty$ and thus  \eqref{evp} has an exponentially decaying at both $+\infty$ and $-\infty$ solution if an only if $\lambda$ is such that the matrix $W_+(\lambda)$ in \eqref{defUplus3} has a nonzero kernel,   that is, ${\det}_{d\times d}\big( \cosh(2\ell b)+\frac12(\lambda b)^{-1}(\lambda^2+b^2)\sinh(2\ell b)\big)=0$. As we will see in a second, the latter determinant plays the role of the Jost function. \hfill$\Diamond$\end{remark}
  
 To introduce the Jost function, we first define the  second Jost solution, $U_-(\cdot,\lambda)$,  to \eqref{evp} which is asymptotic to the decaying plane wave $e^{\lambda x}I_{d\times d}$ as $x\to-\infty$ provided $\re  \, (\lambda)>0$.
 Clearly, one can set $U_-(x,\lambda)=U_+(-x,\lambda)$ for all $x\in\RR$. On the other hand, if $\re  \, (\lambda)<0$ then $U_+(x,\lambda)$, respectively, $U_-(x,\lambda)$, exponentially grows as $x\to+\infty$, respectively, $x\to-\infty$.

 We now explain how to remove assumption \eqref{lambdasp} and show that the Jost solutions can be defined for all $\lambda\in\CC$ and are entire in $\lambda\in\CC$. Indeed, although the function $\lambda\mapsto b(\lambda)$ is not analytic (and is not even defined for $\lambda$ such that $\lambda^2\in\spec(a)$), the Jost solutions are expressed via the functions $b\mapsto\cosh((x-\ell)b)$, $b^{-1}\sinh((x-\ell)b)$ which are {\em even} in $b$. Thus, $U_+(x,\lambda)$ depends, in fact, not on $b(\lambda)$ but on $(b(\lambda))^2=\lambda^2I_{d\times d}-a$; the latter function is entire in $\lambda$.
 
 To give a more formal argument, we  recall that 
 \begin{equation}\label{defZ}
 \cosh Z = \sum_{n=0}^\infty\frac{(Z^2)^n}{(2n)!} \text{ and }
 Z^{-1}\sinh Z=\sum_{n=0}^\infty\frac{(Z^2)^n}{(2n+1)!}\end{equation}
 for any invertible matrix $Z\in\CC^{d\times d}$, where the series converge absolutely in the matrix norm. In particular, when $Z$ is not invertible, the right hand side of the second formula makes perfect sense and thus provides  the definition of the matrix $Z^{-1}\sinh Z$ in the left-hand side without assuming invertibility\footnote{If this is the case then the part of $Z^{-1}\sinh Z$ on $\ker (Z)$ is equal to the identity operator on the subspace.} of $Z$. Since the function $\lambda\mapsto (b(\lambda))^2=\lambda^2I_{d\times d}-a$ is entire, it thus follows that the functions $\lambda\mapsto\cosh(b(\lambda)(x-\ell))$,
 $\lambda\mapsto b^{-1}\sinh( 2\ell b(\lambda))$ and $\lambda\mapsto\lambda^{-1}\sinh(\lambda(x+\ell))$ in the right-hand side of \eqref{defUplus} for each fixed $x\in\RR$ could be continued from the set of $\lambda$ satisfying \eqref{lambdasp} to the whole of  $\CC$ as entire functions of $\lambda$ with values in $\CC^{d\times d}$. We conclude that the Jost solutions $U_\pm(x,\lambda)$ for each $x\in\RR$ are entire functions of $\lambda$.
 
 \begin{remark}\label{remb0}
 We now describe in more details changes to be made in formulas \eqref{defUplus}, \eqref{defUplus2} and \eqref{defUplus3} at the point $\lambda=\lambda_0$ such that $\lambda_0^2\in\spec(a)$, that is, when the matrix $b(\lambda_0)=\sqrt{\lambda^2_0I_{d\times d}-a}$ is not invertible. We let $P_0=P_0(\lambda_0)$ denote the Riesz projection for $b(\lambda_0)$ in $\CC^d$ onto the zero eigenvalue of $b(\lambda_0)$; in particular,
 $\ker(b(\lambda_0))=\ran(P_0)$. We denote by $b_0=b_0(\lambda_0)$ the restriction of $b(\lambda_0)$ onto $\ker (P_0)$  
  such that $b(\lambda_0)=b_0(I_{d\times d}-P_0)$ and view $b_0$ 
  is an invertible operator acting in the $\big(d-\dim(\ker(b(\lambda_0))\big)$-dimensional subspace  $\ker(P_0)$. Formulas \eqref{defZ} now show that 
  $\cosh(sb)=P_0+\cosh(sb_0)(I_{d\times d}-P_0)$ and $b^{-1}\sinh(sb)=sP_0+
  b^{-1}_0\sinh(sb_0)(I_{d\times d}-P_0)$ for any $s\in\RR$. Thus, formulas \eqref{defUplus}, \eqref{defUplus2} and \eqref{defUplus3} hold at $\lambda=\lambda_0$ provided the following changes are made: $\cosh((x-\ell)b)$ and $\cosh(2\ell b)$
 are replaced by $P_0
  +\cosh(b_0(x-\ell))(I_{d\times d}-P_0)$ and $P_0+\cosh(2\ell b_0)(I_{d\times d}-P_0)$ while
  $b^{-1}\sinh((x-\ell)b)$ and $b^{-1}\sinh(2\ell b)$ are replaced by
  $(x-\ell)P_0+b^{-1}_0\sinh(b_0(x-\ell))(I_{d\times d}-P_0)$ and $2\ell P_0+b^{-1}_0\sinh(2\ell b_0)(I_{d\times d}-P_0)$. In particular, if $d=1$ and $\lambda_0^2=a$ we set $P_0=I$ and $b_0=0$ so that $\cosh((x-\ell)b)$ and $\cosh(2\ell b)$ are replaced by $1$ while while
  $b^{-1}\sinh((x-\ell)b)$ and $b^{-1}\sinh(2\ell b)$ are replaced by
  $x-\ell$ and $2\ell$ respectively.
 \hfill$\Diamond$\end{remark}
Alternatively, as it is always done in scattering theory, see, e.g.,  \cite[Section 4]{GM}, one can introduce the Jost solutions $F_\pm$ as solutions to the Volterra equations
\begin{equation}\label{Volterra-Jost}
F_\pm(x,\lambda)=e^{\mp\lambda x}I_{d\times d}+\lambda^{-1}\int_{x}^{\pm\infty}\sinh{(\lambda(x-y))}V(y)F_{\pm}(y,\lambda)\rmd y,\, x\in\RR, \lambda\in\CC;
\end{equation}
here the potential $V$ is defined in \eqref{defV} and we note in passing that equation \eqref{Volterra-Jost} makes perfect sense at $\lambda=0$ by \eqref{defZ}.  The Volterra integral operator has a semi-separable kernel, and thus a large body of well-known results applies, cf.\  \cite{GGK1, GM}. In particular,
as shown in \cite[Section 4]{GM},  the integral kernel of the Volterra operators in \eqref{Volterra-Jost} satisfies \cite[Equations (2.5)-(2.6)]{GM} and so using \cite[Theorem 2.4]{GM} or, alternatively, \cite[Section IX.2]{GGK1} or \cite[Sections XIII5, XIII6]{GGK2},
we infer that equations \eqref{Volterra-Jost} have unique solutions. 
Differentiating twice, one can readily check that $F_\pm(\cdot,\lambda)$ satisfies equation \eqref{evp}. In addition, $F_\pm(x,\lambda)=e^{\mp\lambda x}I_d$, whenever $\pm x\geq\ell$. We conclude that, in fact, for any $\lambda\in\CC$ one has $F_\pm(\cdot,\lambda)=U_\pm(\cdot,\lambda)$ for the Jost function $U_+(\cdot,\lambda)$ as defined in \eqref{defUplus}. This property of the Jost solutions will be useful in proving \eqref{W-Jost-connection} and Lemma~\ref{l5.1}.

For any $\lambda\in\CC$ we now define a matrix $W(\lambda)$ (build analogously to the Wronskian) by 
\begin{equation}\label{defW}
W(\lambda)=
U_+(x,\lambda)\d_x (U_-(x,\lambda))-(\d_x U_+(x,\lambda)) U_-(x,\lambda),
\end{equation}
which is $x$-independent because $U_\pm(\cdot,\lambda)$ are solutions to \eqref{evp}. 
Assuming $\lambda^2\notin\spec(a)$, computing $W(\lambda)$ at $x=0$ and using $U_-(x,\lambda)=U_+(-x,\lambda)$ and $W(\lambda)=-2U_+(0,\lambda)(\d_xU_+)(0,\lambda)$ yields 
\begin{equation}\label{Wcomp}\begin{split}
W(\lambda)
=b^{-1}e^{-2\lambda\ell}\big(2\lambda b\cosh(2\ell b)+(b^2+\lambda^2)\sinh(2\ell b)\big)
\text{ for } \lambda^2\notin\spec(a);
\end{split}
\end{equation}
we used the second line in \eqref{defUplus}, its $x$-derivative at $x=0$ and  $b^2+\lambda^2I_{d\times d}=2\lambda^2I_{d\times d}-a$. 

As in Remark \ref{remb0}, formula \eqref{Wcomp} holds without assuming that $b$ is invertible provided $\cosh(2\ell b)$, respectively, $b^{-1}\sinh(2\ell b)$  is replaced 
by $P_0+\cosh(2\ell b_0)(I_{d\times d}-P_0)$, respectively, by $2\ell P_0+ (b_0)^{-1}\sinh(2\ell b_0)(I_{d\times d}-P_0)$. As a, result
 $\lambda\mapsto W(\lambda)$ extends to an entire matrix valued function and so is $\lambda\mapsto\det_{d\times d}W(\lambda)$. In particular, the set of $\lambda\in\CC$ where the matrix $W(\lambda)$ is not invertible can have an accumulation point only at infinity. Also, equation \eqref{Wcomp} shows that $W(\lambda)={2\lambda}W_+(\lambda)$ for the matrix $W_+(\lambda)$ from \eqref{defUplus3}
and so if $\lambda\neq0$ then ${\det}_{d\times d}W(\lambda)=0$ if an only if ${\det}_{d\times d}W_+(\lambda)=0$. 

Formula \eqref{Wcomp} can be also re-written as
\begin{equation}\label{Wcomp2}
W(\lambda)=2\lambda e^{-2\lambda\ell}\big(\cosh(2\ell b)+\frac12\big(\lambda^{-1}b+\lambda b^{-1}\big)\sinh(2\ell b)\big).
\end{equation}
In particular, $W(0)/2=b(0)\sinh(2\ell b(0))/2$, with $b(0)=\sqrt{-a}$, is the residue at $\lambda=0$ of the function $\lambda\mapsto(2\lambda)^{-1}W(\lambda)$ which is analytic in $\CC\setminus\{0\}$. We note that $0\in\spec(W(0))$ if and only if $\alpha:=(\pi n/2\ell)^2\in\spec(a)$ for some $n\in\ZZ$, cf.\ Lemma \ref{lzero} below.

We will call ${\mathcal F}(\lambda):=\det_{d\times d}(W(\lambda))$ the \textit{Jost function}. In particular, if $\re\lambda>0$ then $\lambda^2\in\spec_{\rm d}(A)$ if and only if $\cF(\lambda)=0$, cf.\ Remark \ref{remJsol}. If $a=0$ (and so $V=0$) then the respective matrix in \eqref{defW}  is equal to $2\lambda I_{d\times d}$. It is therefore natural for any $a\in\CC^{d\times d}$ to normalize the ``Wronskian'' $W(\lambda)$  further by dividing it by $2\lambda$. Consequently, quite often (and in the scalar case $d=1$) the Jost function is defined as $\det\big((2\lambda)^{-1}W(\lambda)\big)$, cf., e.g., \cite[Chapter 17]{CS} and see Theorem \ref{t5.2} below.

In the sequel we will use the following identity satisfied by the matrix-values function $W(\cdot)$ and the Jost solutions $U_\pm(\cdot,\cdot)$,
\begin{equation}\label{W-Jost-connection}
W(\lambda)=2\lambda I_{d\times d}-\int_{-\infty}^{+\infty}e^{\pm\lambda x}V(x)U_\pm(x,\lambda)\rmd x\;\mbox{for any}\;\lambda\in\CC;	
\end{equation}	
the integration here is in fact over $[-\ell,\ell]$. To prove \eqref{W-Jost-connection}, one can either evaluate the integral in the right hand side of \eqref{W-Jost-connection} using  \eqref{defUplus}  (a simple but somewhat tedious computation), or, alternatively, differentiate \eqref{Volterra-Jost} and then plug in $x=\pm\ell$ to \eqref{defW}.  

Recalling that $a\in\CC^{d\times d}$, we now define, for all $\lambda\in\CC$, the Greens function
\begin{equation}\label{defG}
G_\lambda(x,y)=\begin{cases}W(\lambda)^{-1}U_+(x,\lambda)U_-(y,\lambda), & \text{ for } x>y,\\ W(\lambda)^{-1}
U_+(y,\lambda)U_-(x,\lambda), &\text{ for } x<y,\end{cases}
\end{equation}
stressing that the matrices $W(\lambda)$, $U_\pm(x,\lambda)$ commute. Then, for each $x,y\in\RR$, the function $\lambda\mapsto G_\lambda(x,y)$ is meromorphic in $\CC$ with the poles at the points $\lambda$ where $\det_{d\times d}(W(\lambda))=0$.

If $\re  \,(\lambda)>0$ then the Jost solutions decay at infinity exponentially as we have mentioned above. This implies that the resolvent of $A$ at $\lambda^2$ for $\re  \,(\lambda)>0$ such that ${\mathcal F}(\lambda)\neq0$ is given by the formula 
\begin{align}
\big((&\lambda^2-A)^{-1}f\big)(x)=\int_{-\infty}^{+\infty}G_\lambda(x,y)f(y)\, dy\label{resA}\\
&=(W(\lambda))^{-1}U_-(x,\lambda)\int_{-\infty}^x U_+(y,\lambda)f(y)\,dy+
(W(\lambda))^{-1}U_+(x,\lambda)\int_{-\infty}^x U_-(y,\lambda)f(y)\,dy\nonumber
\end{align}
for each $f\in L^2(\RR; \CC^d)$. Indeed, to check that the integral operator with the kernel $G_\lambda(x,y)$ is bounded in $L^2(\RR; \CC^d)$ we rewrite \eqref{defUplus} as $U_+(x,\lambda)=c(x,\lambda,\ell)e^{-\lambda x}$ where  $c(x,\lambda,\ell)\in\CC^{d\times d}$ are such that 
$\|c(\cdot,\lambda,\ell)\|_{L^\infty(\RR)}<\infty$ for each $\lambda$ with $\re  (\lambda)>0$, ${\mathcal F}(\lambda)\neq0$. This, \eqref{defG}, and $U_-(x,\lambda)=U_+(-x,\lambda)$ imply $\|G_\lambda(x,y)\|\le c e^{-\re  \,(\lambda)|x-y|}$ for some computable constant $c>0$ and all $x,y\in\RR$. Using Young's inequality for convolutions, the integral operator in \eqref{resA} is bounded in $L^2(\RR; \CC^d)$.
To check that the sum of the  integrals in the second line of \eqref{resA} indeed gives the resolvent, we denote the right hand side of the equation by $g(x)$ and differentiate $g$ twice to check that $\partial_{xx}g+(V(x)-\lambda^2)g=f$ using that $U_\pm(\cdot,\lambda)$ solve \eqref{evp}.

Formula \eqref{resA} could be extended for $\re  \,(\lambda)\le 0$ in the following sense. As we have mentioned above, $U_\pm(\cdot,\lambda)$ are entire in $\lambda$ 
and thus $G_\lambda(x,y)$ is meromorphic in $\lambda$ in $\CC$ with the poles at zeros of ${\mathcal{F}}(\cdot)={\det}_{d\times d}W(\cdot)$.
For 
$\re  \,(\lambda)\le 0$ the Jost solutions $U_\pm(\cdot,\lambda)$ {\em do not} decay exponentially as $x\to\pm\infty$. As a result, there is no convergence of the integrals in \eqref{resA} for all $f\in L^2(\RR; \CC^d)$. However, the integrals do converge for $f\in L^2_{\textrm{comp}}(\RR;\CC)$, the compactly supported functions. Thus, the resolvent $\lambda\mapsto(\lambda^2-A)^{-1}$ admits a meromorphic extension to the entire complex plane $\CC$ as a bounded operator acting from $L^2_{\textrm{comp}}(\RR; \CC^d)$ into 
$H^2_{\loc}(\RR; \CC^d)$, the space of functions that are locally in the Sobolev space. The poles of the resolvent are the zeros of ${\mathcal{F}}(\cdot)={\det}_{d\times d}W(\cdot)$, that is, the {\em resonances} of $A$ (here, and everywhere in the paper, we follow the terminology of \cite{DZ}). We note that $\lambda^2$ is an eigenvalue of $A$ provided $\lambda$ is a pole of resolvent $\lambda\mapsto(\lambda^2-A)^{-1}$ located in the half-plane of $\lambda\in\CC$ such that $\re  \,(\lambda)>0$. It is possible  that a pole of the meromorphic extension of the resolvent $\lambda\mapsto(\lambda^2-A)^{-1}$ is located in the half-plane of $\lambda\in\CC$ such that $\re  \,(\lambda)<0$ and is such that $\lambda^2$ is not necessarily an eigenvalue of $A$. 

It is therefore important to describe the zeros of ${\mathcal{F}}(\cdot)$ located in both half planes; this is our next objective. To this end, we will need two formulas for $W(\lambda)$.  The first formula is obtained by re-writing $W(\lambda)$ in \eqref{Wcomp} as the following product,
\begin{equation}\label{WprodS}
W(\lambda)=2e^{-2\lambda\ell}\big(\lambda\cosh(\ell b)+b\sinh(\ell b)\big)
\big(\cosh(\ell b)+\lambda b^{-1}\sinh(\ell b)\big) \text{ for $\lambda^2\notin\spec(a)$}.
\end{equation}
Acting as in Remark \ref{remb0}, we also compute
\begin{equation}\label{WprodS0}
W(\lambda_0)=2\lambda_0 e^{-2\lambda_0\ell}(1+\lambda_0\ell)P_0+W_0(\lambda_0)(I_{d\times d}-P_0) \text{ for $\lambda_0^2\in\spec(a)$;}
\end{equation}
here and in what follows we denote by $W_0(\lambda_0)$ the operator acting in the subspace $\ker(P_0)$ and obtained by replacing $b$ in \eqref{WprodS} by $b_0$ and $\lambda$ by $\lambda_0$. Formula \eqref{WprodS} shows that the matrix $W(\lambda)$ is not invertible, that is, $\lambda$ is a zero of the Jost function $\cF(\cdot)$, if and only if one of the following two spectral inclusions holds, provided $\lambda^2\notin\spec(a)$,
\begin{align}\label{WprodS1a}
0&\in\spec\big(\lambda\cosh(\ell b)+b\sinh(\ell b)\big),\\
0&\in\spec\big(\cosh(\ell b)+\lambda b^{-1}\sinh(\ell b)\big). \label{WprodS1b}
\end{align}
Formula \eqref{WprodS0} shows that if $\lambda_0^2\in\spec(a)$, that is, if $0\in\spec(b(\lambda_0))$, then $\lambda_0$ is a zero of $\cF(\cdot)$ if and only if one of the following three spectral inclusions hold,
\begin{align}\label{WprodS10a}
0&\in\spec_{\ker(P_0)}\big(\lambda_0\cosh(\ell b_0)+b_0\sinh(\ell b_0)\big),\\
0&\in\spec_{\ker(P_0)}\big(\cosh(\ell b_0)+\lambda_0 b_0^{-1}\sinh(\ell b_0)\big), \, \text{ or } 1+\lambda_0\ell=0. \label{WprodS10b}
\end{align}
In formulas \eqref{WprodS10a}, \eqref{WprodS10b} the spectrum is computed in the subspace $\ker(P_0)$ of the dimension $d_0=d-\dim\big(\ker(b(\lambda_0))\big)$; we recall that $P_0$ is the Riesz projection onto the zero eigenvalue of the $d\times d$ matrix $b(\lambda_0)=\sqrt{\lambda_0^2-a}$, and the $d_0\times d_0$ invertible matrix $b_0=b(\lambda_0)\big|_{\ker(P_0)}$ is the compression of $b$ onto $\ker(P_0)$.

To derive the second formula for $W(\lambda)$, an elementary calculation with $\cosh$ and $\sinh$ shows that \eqref{WprodS} can be re-written as
\begin{equation}\label{WprodF}
W(\lambda)=(2b)^{-1}\big((b+\lambda)e^{\ell(b-\lambda)}-(b-\lambda)e^{-\ell(b+\lambda)}\big)\big((b+\lambda)e^{\ell(b-\lambda)}+(b-\lambda)e^{-\ell(b+\lambda)}\big),
\end{equation}
for $\lambda^2\notin\spec(a)$, while if $\lambda_0^2\in\spec(a)$ then $W_0(\lambda_0)$ in \eqref{WprodS0} can be re-written as 
\begin{equation}\label{WprodF0}\begin{split}
W_0(\lambda_0)&=(2b_0)^{-1}\big((b_0+\lambda_0)e^{\ell(b_0-\lambda_0)}-(b_0-\lambda_0)
e^{-\ell(b_0+\lambda_0)}\big)\\ &\qquad \qquad \times\big((b_0+\lambda_0)e^{\ell(b_0-\lambda_0)}+(b_0-\lambda_0)
e^{-\ell(b_0+\lambda_0)}\big).\end{split}
\end{equation}
Formulas for $W(\lambda)$ and $W_0(\lambda_0)$ from \eqref{WprodF} and \eqref{WprodF0} can be further simplified as follows.
Using that $a$ is invertible and that $a=(\lambda-b)(\lambda+b)$, we conclude that both matrices 
 $\lambda+ b$  and $\lambda-b$ are invertible and that $(b-\lambda)(b+\lambda)^{-1}=-a^{-1}(b-\lambda)^2$.
 We can now re-write $W(\lambda)$ from  \eqref{WprodF} as follows, assuming $\lambda^2\notin\spec(a)$,
\begin{equation}\label{Wprod21}
W(\lambda)=(2b)^{-1}e^{-2\ell(b+\lambda)}(b-\lambda)^{-2}\big(\big( (-a)^{1/2}e^{\ell b}\big)^2-(b-\lambda)^2\big)\big(\big( (a)^{1/2}e^{\ell b}\big)^2-(b-\lambda)^2\big).
\end{equation}
A similar formula with $b$ replaced by $b_0$ and $\lambda$ by $\lambda_0$ holds true for $W_0(\lambda_0)$ from \eqref{WprodF0}. 

As a result, we may use \eqref{Wprod21} to conclude that $W(\lambda)$
 is not invertible if and only if $\lambda$ is a solution of one of the following four equations,
\begin{equation}\label{FOUR}
{\det}_{d\times d}\big( (\pm a)^{1/2}e^{\ell b}\pm(b-\lambda)\big)=0, \text{ provided }\lambda^2\notin\spec(a), 
\end{equation}
or, alternatively, $\lambda$ satisfies one of the following four spectral inclusions,
\begin{equation}\label{FOUR1}
0\in\spec\big( (\pm a)^{1/2}e^{\ell b}\pm(b-\lambda)\big), 
\text{ provided }\lambda^2\notin\spec(a), 
\end{equation}
while for $\lambda_0$ such that $\lambda_0^2\in\spec(a)$ we conclude that 
 $W(\lambda_0)$ is not invertible if and only if one of the following five assertions hold,
\begin{equation}\label{FOUR2}
0\in\spec\big((\pm a)^{1/2}e^{\ell b_0}(I_{d\times d}-P_0)\pm(b_0-\lambda_0)(I_{d\times d}-P_0)\big) \,  \text{ or } (1+\lambda_0\ell)=0.
\end{equation}
In each of the equations in \eqref{FOUR}, \eqref{FOUR1}  and \eqref{FOUR2}
pluses and minuses are allowed to change independently. 

\begin{example}\label{ex:svc} The representations \eqref{Wprod21} could be particularly useful in the scalar $d=1$ case when invertibility of $a$ is equivalent to $a\neq0$ and of $b$ is equivalent to $\lambda^2\neq a$ for the given complex number $a\in\CC$. Under this condition $W(\lambda)=0$ is and only if $\lambda$ satisfies one of the following four equations, 
\begin{equation}\label{sceq}
\sqrt{-a}e^{\ell b(\lambda)}=\pm(b(\lambda)-\lambda),
\sqrt{a}e^{\ell b(\lambda)}=\pm(b(\lambda)-\lambda), \text{ where } b(\lambda)=\sqrt{\lambda^2-a},
\end{equation} which is a particular case of \eqref{FOUR}.\hfill$\Diamond$
\end{example}

We continue the general discussion when $a\in\CC^{d\times d}$, $d\ge1$, and offer yet another way of finding equations for the roots of $\cF(\lambda):=\det(W(\lambda))$ based on the the spectral mapping theorem saying that $\spec(f(Z))=f(\spec(Z))$ 
for a continuous on the spectrum of $Z$ function $f$ and an arbitrary matrix $Z$.  
Indeed, we may use the spectral mapping theorem for  
the matrix $Z =\ell b$ to conclude from \eqref{WprodS1a}, \eqref{WprodS1b} that 
$\lambda\in\CC$ such that $ \lambda^2\notin\spec(a)$ is a root of the Jost function $\cF(\cdot)$ (that is, $\lambda$ is a resonance) if and only if there is an eigenvalue $\alpha\in\spec(a)$ such that $\lambda$ satisfies one of the following two equations,
\begin{align}\label{Wa}
&\lambda\cosh\big(\ell\sqrt{\lambda^2-\alpha}\big)+(\sqrt{\lambda^2-\alpha})\sinh\big(\ell\sqrt{\lambda^2-\alpha}\big)=0,\\
&\cosh\big(\ell\sqrt{\lambda^2-\alpha}\big)+\lambda(\sqrt{\lambda^2-\alpha})^{-1}\sinh\big(\ell\sqrt{\lambda^2-\alpha}\big)=0.
\label{Wb}
\end{align}
Analogously, by \eqref{WprodS10a}, \eqref{WprodS10b},
if $\lambda_0^2\in\spec(a)$ then $\cF(\lambda_0)=0$ (that is, $\lambda_0$ is a resonance)  if and only if
$1+\lambda_0\ell=0$ or there is an $\alpha\in\spec(a)$ such that one of the equations \eqref{Wa}, \eqref{Wb} holds with $\lambda=\lambda_0$.
The same reasoning applied in the spectral inclusions \eqref{FOUR1}, \eqref{FOUR2} shows that $\cF(\lambda)=0$ for any $\lambda\in\CC$ if and only there is a point $\alpha\in\spec(a)$ such that $\lambda$ satisfies one of the following four equations, 
\begin{equation}\label{FOUR3}\begin{split}
&(\textrm{a})\quad \sqrt{\alpha}e^{\ell\sqrt{\lambda^2-\alpha}}=\sqrt{\lambda^2-\alpha}-\lambda,\quad \quad
(\textrm{b})\quad
 \sqrt{\alpha}e^{\ell\sqrt{\lambda^2-\alpha}}=-\sqrt{\lambda^2-\alpha}+\lambda,\\
&(\textrm{c})\quad \sqrt{-\alpha}e^{\ell\sqrt{\lambda^2-\alpha}}=\sqrt{\lambda^2-\alpha}-\lambda,
\quad  (\textrm{d})\quad \sqrt{-\alpha}e^{\ell\sqrt{\lambda^2-\alpha}}=-\sqrt{\lambda^2-\alpha}+\lambda,\end{split}
\end{equation}
or the equation $1+\lambda_0\ell=0$ (when $\lambda^2_0\in\spec(a)$).

Thus, the resonances of the operator $A$ with the matrix valued potential are described as roots of the scalar equations \eqref{Wa}, \eqref{Wb} or \eqref{FOUR3}, that is, the matrix valued case for $a$ has been essentially reduced to the scalar valued case for $\alpha\in\spec(a)$, cf.\ Example \ref{ex:svc}.
We will use the following terminology and notation.

\begin{definition}\label{def:gener} We say that an eigenvalue  $\alpha\in\spec(a)$ {\em generates} a resonance $\lambda=\lambda(\alpha)$ or that $\lambda$ is an {\em $\alpha$-resonance} if 
$\lambda$ satisfies one of the equations \eqref{Wa}, \eqref{Wb} or \eqref{FOUR3}. The set of $\alpha$-resonances is denoted by $\CalR_\alpha$. The set of all resonances of $A$ is denoted by $\CalR(A)$ so that $\CalR(A)=\cup_{\alpha\in\spec(a)}\CalR_\alpha$. \end{definition}

We denote $A_\alpha=\partial_{xx}+\alpha\chi_{[-\ell,\ell]}$, where $\chi_{[-\ell,\ell]}$ is the characteristic function.
In particular, it follows from the previous discussion that $\CalR_\alpha=\CalR(A_\alpha)$, the set of resonances of the operator $A_\alpha$, because $\lambda\in\CalR(A_\alpha)$ if and only if $\lambda$ is a solution of one of the equations \eqref{FOUR3} or the equation $1+\lambda_0\ell=0$ (when $\lambda_0^2=\alpha$). 

We will now discuss several general properties of the equations \eqref{Wa}, \eqref{Wb} starting with two special values of $\lambda$.
\begin{lemma}\label{lzero} Assume $0\notin\spec(a)$. Then $\lambda=0$ is a root of \eqref{Wa}, respectively, \eqref{Wb} with some $\alpha\in\spec(a)$ if and only if 
$\alpha:=(\pi n/\ell)^2\in\spec(a)$ for some $n\in\ZZ\setminus\{0\}$, respectively, 
$\alpha:=(\pi/(2\ell)+\pi n/\ell)^2\in\spec(a)$ for some $n\in\ZZ$. If this is the case, then $\lambda=0$ is a simple root of the respective equation.
\end{lemma} 

\begin{proof} Setting $\lambda=0$ in \eqref{Wa} yields $\sqrt{-\alpha}\sinh(\ell\sqrt{-\alpha} )=0$. Since ${\alpha}\neq0$ by the assumption $0\notin\spec(a)$, we have $\alpha=(\pi n/\ell)^2\in\spec(a)$ for some $n\in\ZZ\setminus\{0\}$ as required. Conversely, if 
$\alpha:=(\pi n/\ell)^2\in\spec(a)$ for some $n\in\ZZ\setminus\{0\}$ then $\sinh(\ell\sqrt{-\alpha} )=0$ and so $\lambda=0$ solves \eqref{Wa} since $\cosh(\ell\sqrt{-\alpha} )\neq0$ with this choice of $\alpha$. If  $\alpha=(\pi n/\ell)^2\in\spec(a)$ for some $n\in\ZZ\setminus\{0\}$ then the $\lambda$-derivative of the left hand side of \eqref{Wa} is not zero at $\lambda=0$ and so the root is simple. We deal with \eqref{Wb} analogously.
\end{proof}
The second special case is when $\lambda^2\in\spec(a)$.
\begin{lemma}\label{l1l} Assume  $0\notin\spec(a)$ and $\lambda^2\in\spec(a)$. Then \eqref{Wa} does not hold with $\alpha:=\lambda^2$ while
\eqref{Wb} holds with $\alpha:=\lambda^2$ if and only if $\lambda=-1/\ell$.
In particular, if $1/\ell^2\in\spec(a)$ then $\lambda=-1/\ell$ is a simple root of the function in the left hand side of \eqref{Wb}.\end{lemma}
\begin{proof} Setting $\lambda^2=\alpha$ in \eqref{Wa} yields $\lambda=0$ which contradicts the assumptions $0\notin\spec(a)$  and $\lambda^2\in\spec(a)$, thus finishing the proof for \eqref{Wa}. To deal with \eqref{Wb}, by continuity at $z=0$ of the function $z^{-1}\sinh z$, setting $\lambda^2=\alpha$ in \eqref{Wb} yields $1+\lambda\ell=0$ as required. Conversely, if $\lambda=-1/\ell$ and $\lambda^2\in\spec(a)$ by the assumption,  then $\alpha:=\lambda^2\in\spec(a)$ as required. To see the last assertion in the lemma, we introduce a new variable $\nu$ related to $\lambda$ via $\nu=\sqrt{(\ell\lambda)^2-1}$ and $\lambda=-(1/\ell)\sqrt{\nu^2+1}$ and re-write the function $f(\lambda)$ in the left hand side of \eqref{Wb} as $f(\lambda)=g(\nu(\lambda))$ where $g(\nu)=\cosh \nu-\sqrt{\nu^2+1}\sinh \nu/\nu$. Computing  the Taylor series for $g$ at $\nu=0$ up to $\nu^2$ yields $f(-1/\ell)=0$ but $f'(-1/\ell)\neq0$.
\end{proof}

It is sometimes convenient to re-write equations \eqref{Wa}, \eqref{Wb} using $\tanh$  (assuming that $\lambda$ is not as discussed in the two special cases above).
\begin{lemma}\label{tcl}
Assume that $0\notin\spec(a)$, $\lambda\neq0$ and $\lambda^2\notin\spec(a)$. Then $\lambda$ is a solution to \eqref{Wa} with some $\alpha\in\spec(a)$ if and only if  conditions 
\begin{equation}\label{wwab}
\cosh\big(\ell\sqrt{\lambda^2-\alpha}\big)\neq0 \text{ and } \sinh\big(\ell\sqrt{\lambda^2-\alpha}\big)\neq0
\end{equation}
hold and $\lambda$ solves the equation
\begin{equation}\label{wwa} (\sqrt{\lambda^2-\alpha})\lambda^{-1}\tanh\big(\ell\sqrt{\lambda^2-\alpha}\big)+1=0
\end{equation}
 while $\lambda$ is a solution to \eqref{Wb} with some $\alpha\in\spec(a)$ if and only if
conditions \eqref{wwab} hold and $\lambda$ solves the equation
\begin{equation}\label{wwb}  (\sqrt{\lambda^2-\alpha})^{-1}\lambda \tanh\big(\ell\sqrt{\lambda^2-\alpha}\big)+1=0.
\end{equation}
In other words, $\lambda\in\CC$ is a resonance, $\cF(\lambda)=0$, if and only if \eqref{wwa} or \eqref{wwb} are satisfied.
\end{lemma}
\begin{proof} We recall that $\cosh z=0$, respectively, $\sinh z=0$ if and only if $z=\rmi(\pi/2+\pi n)$, respectively,  $z=\rmi\pi n$, for some $n\in\ZZ$. Suppose that $\lambda$ solves \eqref{Wa} but $\cosh\big(\ell\sqrt{\lambda^2-\alpha}\big)=0$. Then $\sinh\big(\ell\sqrt{\lambda^2-\alpha}\big)\neq0$ and thus \eqref{Wa} yields $\sqrt{\lambda^2-\alpha}=0$ contradicting the assumption $\lambda^2\notin\spec(a)$ in the lemma. Suppose that $\lambda$ solves \eqref{Wa} but $\sinh\big(\ell\sqrt{\lambda^2-\alpha}\big)=0$. Then $\cosh\big(\ell\sqrt{\lambda^2-\alpha}\big)\neq0$ and thus \eqref{Wa} yields $\lambda=0$ contradicting the assumption $\lambda\neq0$ in the lemma.
Thus, conditions \eqref{wwa} do hold. Using this and the assumption $\lambda\neq0$ in the lemma, we derive from \eqref{Wa} assertion \eqref{wwa}. Conversely, if \eqref{wwab} and \eqref{wwa} hold then \eqref{Wa} follows using the assumptions $\lambda\neq0$ and $\lambda^2\notin\spec(a)$ in the lemma. The proof for \eqref{wwab}, \eqref{Wb} and \eqref{wwb} is analogous.
\end{proof}

\section{Resonances generated by complex eigenvalues}\label{Complex-Resonances}
As we have mentioned above, the values of $\lambda$ such that ${\mathcal{F}}(\lambda)={\det}_{d\times d}W(\lambda)=0$ are the poles of the Green's function $\lambda\mapsto G_\lambda$ and therefore the poles of the meromorphic extension of the resolvent of $A$, that is, the {\em resonances}, cf., e.g., \cite{DZ}. In this section we study the roots of equations \eqref{FOUR3} for the general case of a complex eigenvalue $\alpha\in\spec(a)$. First, we make a change of variables to simplify the equations. Equations \eqref{FOUR3} (a)-(b), respectively,  \eqref{FOUR3} (c)-(d), are equivalent to, respectively,
\begin{align}\label{equiv-ab}
\alpha e^{2\ell\sqrt{\lambda^2-\alpha}}&=(\lambda-\sqrt{\lambda^2-\alpha})^2,\\
\label{equiv-cd}
-\alpha e^{2\ell\sqrt{\lambda^2-\alpha}}&=(\lambda-\sqrt{\lambda^2-\alpha})^2.
\end{align}
We make the change of variables $\lambda=\sqrt{\alpha}\cosh{z}$ in equations \eqref{equiv-ab}-\eqref{equiv-cd}. A simple computation shows that 
$\sqrt{\lambda^2-\alpha}=\sqrt{\alpha}\sinh{z}$ if $\re(\sqrt{\alpha}\sinh{z})\geq0$ and  $\sqrt{\lambda^2-\alpha}=-\sqrt{\alpha}\sinh{z}$ if  $\re(\sqrt{\alpha}\sinh{z})<0$, respectively, and so \eqref{equiv-ab}
becomes $e^{2\ell\sqrt{\alpha}\sinh{z}}=e^{-2z}$ and
$e^{-2\ell\sqrt{\alpha}\sinh{z}}=e^{2z}$, respectively.
Since the last two equations are equivalent, equation \eqref{equiv-ab} is equivalent to 
\begin{equation}\label{res-ab} 
\lambda=\sqrt{\alpha}\cosh{z}\;\mbox{ and }\;\ell\sqrt{\alpha}\sinh{z}+z\in\rmi\pi\ZZ.
\end{equation}
Similarly, equation \eqref{equiv-cd} is equivalent to 
\begin{equation}\label{res-cd} 
\lambda=\sqrt{\alpha}\cosh{z}\;\mbox{and}\;\ell\sqrt{\alpha}\sinh{z}+z\in\rmi\pi\ZZ+\rmi{\pi}/{2}.
\end{equation}
We are now ready to present the first main result of this section.
\begin{theorem}\label{t3.1} Assume that $0\notin\spec(a)$ and let $\alpha\in\CC$ be an eigenvalue of $a$. Then there are infinitely many resonances generated by $\alpha$.
\end{theorem}
\begin{proof} 
We introduce $f_\alpha:\CC\to\CC$ by $f_\alpha(z)=\ell\sqrt{\alpha}\sinh{z}+z$ and note that $f_\alpha$ is an entire function that has an essential singularity at infinity. From Picard's Great Theorem we know that there exists $\nu_0\in\CC$ such that the set
\begin{equation}\label{t3.1.1}
\{z\in\CC: f_\alpha(z)=\nu\}\;\mbox{is infinite for any}\; \nu\in\CC\setminus\{\nu_0\}.	
\end{equation}	   
However, one can readily check that 
\begin{equation}\label{t3.1.2}
f_\alpha(z+2\pi\rmi)=f_\alpha(z)+2\pi\rmi\;\mbox{for any}\;z\in\CC,
\end{equation}	
which implies that 
\begin{equation}\label{t3.1.3}
\{z\in\CC: f_\alpha(z)=\nu_0\}=\{z\in\CC: f_\alpha(z)=\nu_0+2\pi\rmi\}-2\pi\rmi.
\end{equation}
This last identity allows us to improve on \eqref{t3.1.1}. Indeed, from \eqref{t3.1.1} and \eqref{t3.1.3} we infer 
\begin{equation}\label{t3.1.4}
\{z\in\CC: f_\alpha(z)=\nu\}\;\mbox{is infinite for any}\; \nu\in\CC.	
\end{equation}
In particular, the function $f_\alpha$ has infinitely many zeros. Next, we define 
\begin{equation}\label{t3.1.5}
Z_\pm(f_\alpha)=\{z\in\CC:f_\alpha(z)=0,\pm\im z>0\}.	
\end{equation}
Since the function $f_\alpha$ is odd we have that 
\begin{equation}\label{t3.1.6}
z\in Z_+(f_\alpha)\;\mbox{if and only if}\;-z\in Z_-(f_\alpha).
\end{equation}
Moreover, one can readily check that $f_\alpha$ has at most three roots on the real line. From \eqref{t3.1.4} and \eqref{t3.1.6} we conclude that $Z_\pm(f_\alpha)$ are both infinite sets. Hence, there exists a sequence $\{z_n\}_{n\in\NN}$ of complex numbers in $Z_+(f_\alpha)$ such that $z_k\ne z_j$ for any $j,k\in\NN$ with $k\ne j$. From \eqref{res-ab} we obtain   $\lambda_n:=\sqrt{\alpha}\cosh{z_n}$, $n\in\NN$, are resonances generated by $\alpha$. To prove the theorem it is enough to show that $\lambda_k\ne\lambda_j$, whenever $k,j\in\NN$ with $k\ne j$. To prove this claim we assume that $\lambda_k=\lambda_j$. It follows that $\cosh{z_k}=\cosh{z_j}$, thus 
\begin{equation}\label{t3.1.7} \text{ Case A:  $z_k-z_j\in2\pi\rmi\ZZ$ and Case B: $z_k+z_j\in2\pi\rmi\ZZ$}.
\end{equation}
We note that in Case A we have $z_k=-l\sqrt{\alpha}\sinh{z_k}=-l\sqrt{\alpha}\sinh{z_j}=z_j$, which implies that $k=j$. In Case B we obtain  $z_k=-l\sqrt{\alpha}\sinh{z_k}=l\sqrt{\alpha}\sinh{z_j}=-z_j$, which is a contradiction with the fact that $z_k,z_j\in Z_+(f_\alpha)$. We conclude that $\lambda_k\ne\lambda_j$, whenever $k,j\in\NN$ with $k\ne j$, proving the theorem.
\end{proof}
In the remaining part of this section we aim to locate the complex resonances generated by the eigenvalue $\alpha$. We will show that there are only finitely many eigenvalues located in the right half plane and will describe the asymptotic of the resonances converging to infinity and located in the left half plane.  To achieve this we first describe the location of the solutions of the equation 
\begin{equation}\label{S0mu} 
\ell\sqrt{\alpha}\sinh{z}+z=\rmi\mu,\;\mbox{with}\;\mu\in\pi\ZZ\cup(\pi\ZZ+\pi/2).
\end{equation}
To begin, we introduce the {\em real} variables
 $p=\ell\re(\sqrt{\alpha})$, $q=\ell\im(\sqrt{\alpha})$, $t=\re  z$, $s=\im z$
 so that $\ell\sqrt{\alpha}=p+\rmi q$ and $z=t+\rmi s$. Since $\alpha\ne0$ we have
\begin{equation}\label{p-q-condition}
p\geq 0,\quad(p,q)\ne(0,0).
\end{equation}
Moreover, for each $\mu\in\pi\ZZ\cup(\pi\ZZ+\pi/2)$ equation \eqref{S0mu} is equivalent to the system
\begin{equation}\label{sys-S0mu}\tag{$S_0(p,q,\mu)$}
\left\{\begin{array}{ll}
p\sinh{t}\cos{s}-q\cosh{t}\sin{s}+t=0,\\
q\sinh{t}\cos{s}+p\cosh{t}\sin{s}+s=\mu.\end{array}\right.
\end{equation} 
\begin{remark}\label{r3.13}
By periodicity, if $(t,s)$ is a solution of the system $S_0(p,q,\mu)$, then $(t,s-2\pi n)$ is a solution of the system $S_0(p,q,\mu-2\pi n)$ for any $n\in\ZZ$.\hfill$\Diamond$
\end{remark}
As shown in Theorem~\ref{t3.1}, equation \eqref{S0mu} has infinitely many solutions in the complex plane,  which implies that system \eqref{sys-S0mu} has infinitely many solutions $(p,q)\in\RR^2$. Moreover,  system \eqref{sys-S0mu} has infinitely many solutions that belong to each quadrant of $\RR^2$ see Lemma \ref{l3.2} below. 
This result, together with two rather technical Lemmas \ref{l3.5} and \ref{l3.8}, is proved in Appendix~\ref{App}. The lemmas also provide the proof of our next result.
\begin{lemma}\label{l3.10}
Assume $p\geq 0$ and $(p,q)\ne (0,0)$, let  $\{\mu_n\}_{n\in\NN}$ be a finite sequence, and $\{(t_n,s_n)\}_{n\in\NN}$ be a sequence of solutions to $(S_0(p,q,\mu_n))$ such that $|t_n|\to\infty$ and $|s_n|\to\infty$ as $n\to\infty$. Then, 
\begin{equation}\label{log-asym}
\lim_{n\to\infty}{|s_n|}{e^{-|t_n|}}=\sqrt{p^2+q^2}/{2}.	
\end{equation}
\end{lemma}	

We will need some more information on the solutions of \eqref{sys-S0mu}.  
\begin{lemma}\label{l3.3}
Assume that $\alpha\in\CC\setminus\{0\}$, and let $\{\mu_n\}_{n\in\NN}\subset\RR$ be a {\em finite} sequence while $\{(t_n,s_n)\}_{n\in\NN}\subset\RR^2$ be an {\em infinite} sequence of solutions of system $(S_0(p,q,\mu_n))$. Then the sequences $\{t_n\}_{n\in\NN}$ and $\{s_n\}_{n\in\NN}$ are unbounded.
\end{lemma}	
\begin{proof}
Seeking a contradiction, let us suppose that $\{t_n\}_{n\in\NN}$ is bounded. Since the finite sequence $\{\mu_n\}_{n\in\NN}$ is bounded, 
  $\{s_n\}_{n\in\NN}$ is also bounded by the second equation in $(S_0(p,q,\mu_n))$. Moreover, since $\{\mu_n\}_{n\in\NN}$ is finite, by passing to a subsequence we will assume that $\mu_{n}=\mu_0$ for some $\mu_0\in\RR$ while the values $(t_n,s_n)$ are distinct for all $n\in\NN$.
 Thus, $\{t_{n}+\rmi s_{n}\}_{n\in\NN}$ is an infinite bounded sequence of zeros of the entire nonzero function $z\to\ell\sqrt{\alpha}\sinh{z}+z-\rmi\mu_0$. The contradiction proves that $\{t_n\}_{n\in\NN}$ is unbounded. Passing to a subsequence we will assume that $t_n\to+\infty$ since the case $t_n\to-\infty$ is analogous.

Seeking a contradiction, let us suppose that $\{s_{n}\}_{n\in\NN}$ is bounded. Passing to a subsequence we will assume that $s_{n}\to s_*$ 
for some $s_*\in\RR$ as $n\to\infty$. 
By the first equation in \eqref{sys-S0mu}, 
\begin{equation*}\label{l3.3.1} 
\big(p\cos{(s_n)}-q\sin{(s_n)}\big)e^{t_n}/(2t_n)-\big(p\cos{(s_n)}+q\sin{(s_n)}\big)e^{-t_n}/(2t_n)+1=0,
\end{equation*}
 while by the second equation in \eqref{sys-S0mu},
\begin{equation*}\label{l3.3.3} 
\big(q\cos{(s_n)}+p\sin{(s_n)}\big)e^{t_n}/2+\big(-q\cos{(s_n)}+p\sin{(s_n)}\big)e^{-t_n}/2+s_n=\mu_n,
\end{equation*}
yielding
\begin{equation}\label{l3.3.4} p\cos{(s_*)}-q\sin{(s_*)}=0 \text{ and }
q\cos(s_*)+p\sin(s_*)=0, 
\end{equation}  
respectively, because $e^{t_n}/(2t_n)\to\infty$, $t_n\to\infty$, $s_n\to s_*$ and $\{\mu_n\}$ is bounded as $n\to\infty$.  
  But \eqref{l3.3.4} contradicts \eqref{p-q-condition}. 
\end{proof}
 Our next simple but very useful lemma describes the position of the infinitely many resonances generated by the eigenvalue $\alpha\in\CC\setminus\{0\}$ of $a$.
\begin{lemma}\label{l3.11}
Assume $\alpha\in\CC\setminus\{0\}$, $\ell\sqrt{\alpha}=p+\rmi q$, and let $\lambda=\sqrt{\alpha}\cosh{(t+\rmi s)}$ where $(t,s)$ solves system $(S_0(p,q,\mu))$ for some $\mu\in\RR$. Then
\begin{align}\label{Real-lambda}
\big|\ell\,\re\lambda+|t|\big|&\leq(|p|+|q|)e^{-|t|},\\	
\label{Imaginary-lambda}
\big|\sgn{(t)}\ell\,\im\lambda+s-\mu\big|&\leq(|p|+|q|)e^{-|t|}.
\end{align}		
\end{lemma}	
\begin{proof}
Rearranging terms in $(S_0(p,q,\mu))$ we have
\begin{align}\label{l3.11.1a}
e^t(p\cos{s}-q\sin{s})-e^{-t}(p\cos{s}+q\sin{s})&=-2t,\\
e^t(q\cos{s}+p\sin{s})+e^{-t}(-q\cos{s}+p\sin{s})&=2\mu-2s.\label{l3.11.1b}\end{align} Taking the real and imaginary parts in $\ell\lambda=(p+\rmi q)\cosh{(t+\rmi s)}$ yields 
\begin{align}\label{l3.11.2a}
2\ell\,\re\lambda&=
e^t(p\cos{s}-q\sin{s})+e^{-t}(p\cos{s}+q\sin{s}),\\
2\ell\,\im\lambda&
=e^t(q\cos{s}+p\sin{s})-e^{-t}(-q\cos{s}+p\sin{s}),\label{l3.11.2b}
\end{align}
which, combined with \eqref{l3.11.1a}, gives identities 
\begin{equation*}\label{l3.11.3}
\ell\,\re\lambda=-t+e^{-t}(p\cos{s}+q\sin{s})\;\text{ and }\;\ell\,\re\lambda=t+e^t(p\cos{s}-q\sin{s}).
\end{equation*}
Using the first identity for $t\geq 0$ and the second identity for $t\leq 0$, we immediately infer \eqref{Real-lambda}. Similarly, from \eqref{l3.11.1b} and \eqref{l3.11.2b} we obtain the identities 
\begin{equation*}\label{l3.11.4}
\ell\,\im\lambda=(\mu-s)-e^{-t}(-q\cos{s}+p\sin{s})\;\mbox{and}\;\ell\,\im\lambda=-(\mu-s)+e^t(q\cos{s}+p\sin{s}),
\end{equation*}
and assertion \eqref{Imaginary-lambda} follows by using the first identity for $t\geq 0$ and the second for $t\leq 0$. 
\end{proof}
A main result of this section is given in the next theorem.
\begin{theorem}\label{t3.12}
Assume that $0\notin\spec(a)$ and let $\alpha\in\CC$ be an eigenvalue of $a$. Then there are finitely many $\alpha$-resonances in the closed right half plane  $\overline{\CC}_+=\{\lambda\in\CC:\re\lambda\geq0\}$.
\end{theorem}
\begin{proof}
Let $\lambda$ be a resonance generated by the eigenvalue $\alpha$, that is, a solution of \eqref{FOUR3}. Then there exist a $\mu\in\pi\ZZ\cup(\pi\ZZ+\pi/2)$ and a solution $(t,s)$ of system $(S_0(p,q,\mu))$ such that $\lambda=\sqrt{\alpha}\cosh{(t+is)}$. If $\re\lambda\geq 0$ then  from \eqref{Real-lambda}  in Lemma~\ref{l3.11}  we infer 
\begin{equation}\label{l3.12.1}
|t|\leq(|p|+|q|)e^{-|t|}.	
\end{equation}
By Lemma~\ref{l3.3} system $(S_0(p,q,\mu))$ may have only finitely many solutions $(t,s)$ with bounded $t$-components. Since the set $\{t\in\RR:|t|\leq(|p|+|q|)e^{-|t|}\}$ is bounded, by \eqref{l3.12.1} we conclude that there may be  only finitely many $\alpha$-resonances in $\overline{\CC}_+$, proving the theorem.	 	
\end{proof}	
 \begin{lemma}\label{l3.7-new}
Assume that $0\notin\spec(a)$. Then, each $\alpha\in\spec(a)$ generates infinitely many resonances located in the second and  infinitely many resonances located in the third quadrant of the complex plane. 
 \end{lemma}
 \begin{proof}
The lemma follows shortly from Lemma~\ref{l3.11} and Lemma \ref{l3.2}.
 \end{proof}
 Our final main result of this section describes the asymptotic behavior of infinitely many $\alpha$-resonances lying in the left  half plane $\CC_-=\{\lambda\in\CC:\re\lambda<0\}$. Results of this type have of course a long history, and are in the spirit of \cite{LaP}, cf.\
also \cite[Theorem 2.10]{DZ}.
\begin{theorem}\label{t3.14}
Assume that $0\notin\spec(a)$ and let $\alpha$ be an eigenvalue of $a$. Then:
\begin{enumerate}
\item[(i)] If $\{\lambda_n\}_{n\in\NN}$ is an infinite  sequence of resonances generated by $\alpha$ then the sequences $\{\re\lambda_n\}_{n\in\NN}$ and $\{\im\lambda_n\}_{n\in\NN}$ are unbounded;
\item[(ii)] If $\{\lambda_n\}_{n\in\NN}$ is an infinite sequence of resonances generated by $\alpha$ such that $\re\lambda_n\to-\infty$ and $|\im\lambda_n|\to\infty$ as $n\to\infty$ then
\begin{equation}\label{asympt-res}
\lim_{n\to\infty}\big({|\im\lambda_n|}{e^{\ell\re\lambda_n}}\big)={|\sqrt{\alpha}|}/{2}>0;
\end{equation}
\item[(iii)] For any $\eta>0$ and $\delta\in\Big(0,{1}/{\ell}\Big)$ there are finitely many resonances generated by the eigenvalue $\alpha$ lying to the right of the curve
$\re\lambda=-\eta-\delta\ln(1+|\im\lambda|)$.	
\end{enumerate}	
\end{theorem}	
\begin{proof}
(i) Since $\lambda_n$, $n\in\NN$, is a resonance generated by $\alpha$, from \eqref{res-ab} and \eqref{res-cd} we infer that there exist a sequence $\{\mu_n\}_{n\in\NN}\subset\RR$  and a sequence $\{(t_n,s_n)\}_{n\in\NN}$ of solutions to $(S_0(p,q,\mu_n))$ such that
\begin{equation}\label{t3.14.1}
\lambda_n=\sqrt{\alpha}\cosh{(t_n+\rmi s_n)}\;\mbox{and}\;\mu_n\in\pi\ZZ\cup(\pi\ZZ+\pi/2)\;\mbox{for any}\;n\in\NN.	
\end{equation}	
We split
\begin{equation}\label{t3.14.2}
\mu_n=2\pi i_n+\tmu_n\, \text{ where $i_n\in\ZZ$, $\tmu_n\in\{0,\pi/2,\pi,3\pi/2\}$ and define $\ts_n=s_n-2\pi i_n$.}		
\end{equation}
 By Remark~\ref{r3.13}
$(t_n,\ts_n)$ is a solution to $(S_0(p,q,\tmu_n))$. Since the sequence $\{\lambda_n\}_{n\in\NN}$ is infinite and $\cosh(\cdot)$ is periodic with period $2\pi\rmi$, by \eqref{t3.14.1}  the sequence $\{(t_n,\ts_n)\}_{n\in\NN}$ is infinite.
By Lemma~\ref{l3.3} the sequences $\{t_n\}_{n\in\NN}$ and $\{\ts_n\}_{n\in\NN}$ are unbounded.  Passing to a subsequence we may assume  $|\ts_{n}|\to\infty$ as $n\to\infty$.
 Using the second equation of $(S_0(p,q,\tmu_{n}))$, 
\begin{equation}\label{t3.14.3}
q\sinh{(t_{n})}\cos{(\ts_{n})}+p\cosh{(t_{n})}\sin{(\ts_{n})}+\ts_{n}=\tmu_{n}.
\end{equation}	
Since the sequence $\{\tmu_n\}_{n\in\NN}$ is finite, by \eqref{t3.14.3}   $\{t_{n}\}_{n\in\NN}$ is unbounded. Again, passing to a
 subsequence, we assume that $|t_{n}|\to\infty$ as $n\to\infty$. Using estimate \eqref{Real-lambda}, we conclude that 
 $\re\lambda_{n}\to-\infty$ as $n\to\infty$, as required. 
 Similarly, by \eqref{Imaginary-lambda}  $\{\sgn{(t_{n})}\,\ell\, \im\lambda_{n}+\ts_{n}\}_{n\in\NN}$ is bounded. Since $|\ts_{n}|\to\infty$  then $|\im\lambda_{n}|\to\infty$ finishing the proof of (i).

\noindent\textit{Proof of (ii).} For an infinite sequence $\{\lambda_n\}_{n\in\NN}$  of $\alpha$-resonances such that $\re\lambda_n\to-\infty$ and $|\im\lambda_n|\to\infty$ as $n\to\infty$ we
let $\{t_n\}_{n\in\NN}$, $\{\ts_n\}_{n\in\NN}$ and $\{\tmu_n\}_{n\in\NN}$ be the sequences defined above in the proof of (i). We begin the proof of formula \eqref{asympt-res} by proving the following two claims.

{\textit{Claim 1.\,}}  $\lim_{n\to\infty}|t_n|=\infty$.
 Indeed, seeking a contradiction and passing to a subsequence we may suppose that
 $|t_{n}|\to r$ as $n\to\infty$ for some $r\in[0,\infty)$. This contradicts  estimate \eqref{Real-lambda} since $\re\lambda_n\to-\infty$ thus proving  Claim 1.
 
{\textit{Claim 2.\,}} $\lim_{n\to\infty}|\ts_n|=\infty$. Indeed, seeking a contradiction and passing to a subsequence we may suppose that $\lim_{n\to\infty}|\ts_n|=r<\infty$. This contradicts estimate 
 \eqref{Imaginary-lambda} since $|\im\lambda_n|\to\infty$ and $\lim_{n\to\infty}|t_n|=\infty$ by Claim 1 thus proving Claim 2.
 
We are ready to prove formula \eqref{asympt-res}. Indeed,  we write
\begin{equation}\label{t3.14.9}
{|\im\lambda_n|}{e^{\ell\re\lambda_n}}=\big({|\im\lambda_n|}/{|\ts_n|}\big)\cdot\big(|\ts_n|{e^{-|t_n|}}\big)\cdot \big(e^{\ell\re\lambda_n+|t_n|)}\big)
\end{equation}
and establish convergence as $n\to\infty$ of each of the three factors in \eqref{t3.14.9} as follows.
By estimate   \eqref{Imaginary-lambda}, Claim 1 and Claim 2 
we have
${\sgn{(t_n)}\,\im\lambda_n}/{\ts_n}\to-{1}/{\ell}$
and therefore ${|\im\lambda_n|}/{|\ts_n|}\to{1}/{\ell}$.
Convergence of the second factor in \eqref{t3.14.9} is proved in Lemma~\ref{l3.10} (and we recall that $\ell\sqrt{\alpha}=p+\rmi q$).
Finally, by \eqref{Real-lambda} and Step 1 we have
$\ell\re\lambda_n+|t_n|\to0$.
Now formula \eqref{asympt-res} and thus (ii) follows by passing to the limit in \eqref{t3.14.9}.	 	

\noindent\textit{Proof of (iii).} Fix $\eta>0$ and $\delta\in(0,{1}/{\ell})$. Seeking a contradiction, suppose there is an infinite  sequence $\{\lambda_n\}_{n\in\NN}$ of $\alpha$-resonances such that
\begin{equation}\label{t3.14.10}
\re\lambda_n>-\eta-\delta\ln{(1+|\im\lambda_n|)}\;\mbox{for any}\;n\in\NN.	
\end{equation}
By (i),  $\{\re\lambda_n\}_{n\in\NN}$ and $\{\im\lambda_n\}_{n\in\NN}$ are unbounded and thus passing to a subsequence we may assume that $|\re\lambda_{n}|\to\infty$ and $|\im\lambda_{n}|\to\infty$	as $n\to\infty$. By Theorem~\ref{t3.12} $\re\lambda_{n}\to-\infty$. Using \eqref{t3.14.10}, 
\begin{equation}\label{t3.14.11}
{|\im\lambda_{n}|}{e^{\ell\re\lambda_{n}}}>e^{-\eta\ell}{|\im\lambda_{n}|}{(1+|\im\lambda_{n}|)^{-\delta\ell}}\;\mbox{for any}\;n\in\NN.	
\end{equation}	
Passing to the limit in \eqref{t3.14.11} and using (ii) and  $\delta\ell<1$ leads to a contradiction.
\end{proof}	

\begin{example}\label{e3.15}
We finish this section with an example illustrating the results on Theorem ~\ref{t3.14}. Let $a=\bigl[\begin{smallmatrix}0&1\\-1&0\end{smallmatrix}\bigr]$ and $\ell=3$. In this case $\spec(a)=\{\pm\rmi\}$. Below are plots of $\rmi$-resonances and $-\rmi$-resonances that show the behavior proved in this section. In particular, we can see that there are finitely many resonances with positive real part, there are infinitely many in the second and third quadrant and the behavior predicted by \eqref{asympt-res} is true.

\begin{figure*}[ht!]
\includegraphics[width=.8\textwidth]{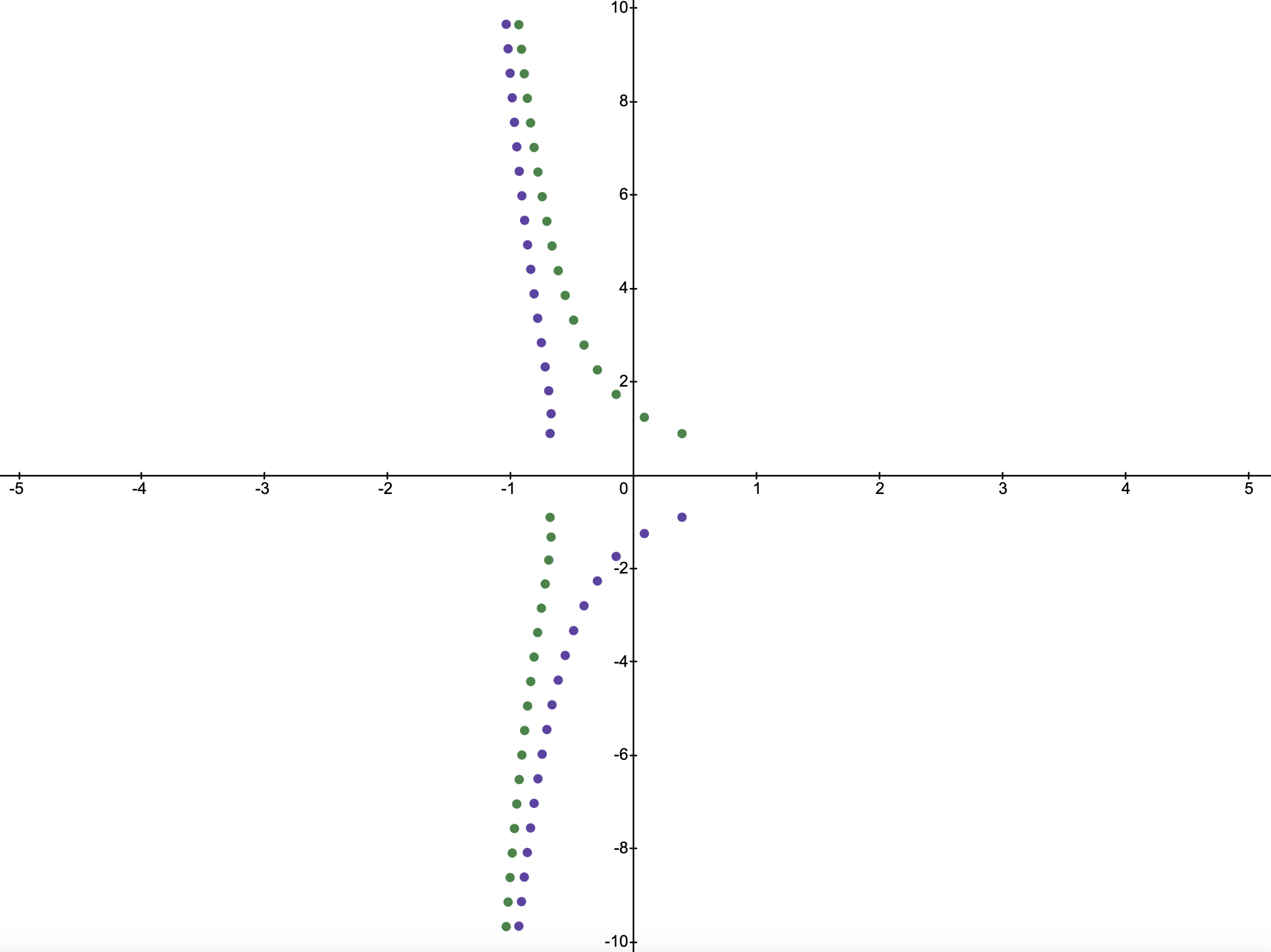}
\caption{A plot of the resonances of $A$ satisfying the condition $|\im\lambda|\leq 10$. The green dots are the $\rmi$-resonances, while the magenta dots are the $-\rmi$ resonances.}
\end{figure*}

\end{example}	

\section{Real and pure imaginary resonances}\label{Real-Zeros}
In this section we will use equations \eqref{wwa}, \eqref{wwb} and \eqref{FOUR3} to describe the resonances located on the real and pure imaginary lines. We begin with the case when the value $\alpha\in\spec(a)$ in equations \eqref{wwa}, \eqref{wwb} and \eqref{FOUR3} is real and, moreover, is positive. 

\begin{theorem}\label{thm: 3.1}  Assume $0\notin\spec(a)$ and let $\alpha$ be a positive eigenvalue of the matrix $a$. 
If $\ell\sqrt{\alpha} =k\pi$ or $\ell\sqrt{\alpha} =\pi/2+k\pi$ for some $k=0,1,\ldots$ then $\lambda=0$ while if $\alpha=1/\ell^2$ then $\lambda=-1/\ell$ is a resonance generated by $\alpha$. 
Assuming $\ell\sqrt{\alpha} \neq k\pi$,
$\ell\sqrt{\alpha} \neq\pi/2+k\pi$, $k\in\ZZ$, and $\alpha\neq 1/\ell^2$, the following holds:
\begin{enumerate}
\item[(i)] There are no resonances generated by $\alpha$ such that $\lambda>\sqrt{\alpha}$.

\item[(ii)] If $\ell\sqrt{\alpha} \in\big(0,1\big)$ then there is at exactly one solution $\lambda$ to \eqref{wwb} such that $\lambda<-\sqrt{\alpha}$ which gives a negative resonance generated by $\alpha$. However, if $\ell\sqrt{\alpha} >1$ then there are no resonances generated by $\alpha$ such that $\lambda<-\sqrt{\alpha}$.

\item[(iii)] If $\ell\sqrt{\alpha} \in\big(0,\pi\big)$ then there is a unique solution $\lambda$ to \eqref{wwa} such that $\lambda\in\big(0,\sqrt{\alpha}\big)$ which gives a positive resonance generated by $\alpha$.

\item[(iv)] Furthermore, if $\ell\sqrt{\alpha} \in\big(k \pi , (k+1)\pi\big)$ for some $k=0,1,\ldots$ then there are  exist $k+1$ solutions $\lambda_j$, $j=0,1,\ldots,k$, to \eqref{wwa} such that \[\lambda_j\in\big((\sqrt{\alpha-(\pi/2+j\pi)^2/\ell^2}),(\sqrt{\alpha-(j\pi)^2/\ell^2}\big)\big), j=0,1,\ldots,k-1, \lambda_k\in\big(0,\sqrt{\alpha-(k\pi)^2/\ell^2}\big),\] each of which gives a positive resonance generated by $\alpha$.

\item[(v)] If $\ell\sqrt{\alpha} \in\big(0,1\big)$ then there are no negative resonances generated by $\alpha$ and such that
$\lambda\in(-\sqrt{\alpha},0)$.

\item(vi)] If $\ell\sqrt{\alpha} \in\big(1,\pi/2\big)$ then there is at least one solution $\lambda$ to \eqref{wwb} such that $\lambda\in\big(-\sqrt{\alpha}, 0\big)$ which gives a negative resonance generated by $\alpha$. 

\item[(vii)] Furthermore, if $\ell\sqrt{\alpha} \in\big(k\pi, (k+1)\pi\big)$ for some $k=1,2,\ldots$ then there are $k$ solutions $\lambda_j$ to \eqref{wwb} such that $\lambda_j\in\big(-(\sqrt{\alpha-(j\pi)^2/\ell^2}), -(\sqrt{\alpha-(\pi/2+j\pi)^2/\ell^2})\big)$, $j=0,1,\ldots,k-1$
each of which gives a negative resonance generated by $\alpha$. 

\item[(viii)] Finally, there are no pure complex resonances $\lambda$ generated by $\alpha$. 
\end{enumerate}
\end{theorem}
\begin{proof}
The first (unnumbered) assertions in the theorem hold by  Lemmas \ref{lzero} and \ref{l1l}. We now consider only the cases when $\ell\sqrt{\alpha} \neq k\pi$,
$\ell\sqrt{\alpha} \neq\pi/2+k\pi$, $k\in\ZZ$ and $\alpha\neq 1/\ell^2$.

We will be looking for real solutions $\lambda\in\RR\setminus\{0\}$ of \eqref{wwa} or \eqref{wwb} and split the proof in the following two cases, 
\begin{equation}\label{AB} \text{ Case A:  $|\lambda|>\sqrt{\alpha}$ and Case B: $|\lambda|<\sqrt{\alpha}$}.
\end{equation}
We begin with Case A in \eqref{AB}.  We set $\beta:=\ell\sqrt{\lambda^2-\alpha}$, note that $\beta>0$ 
and introduce the function $f(\beta)=(\sqrt{\beta^2+\alpha\ell^2})/\beta$.

First, let us assume that $\lambda>\sqrt{\alpha}$ so that $\lambda=(\sqrt{\beta^2+\alpha\ell^2})/\ell>\sqrt{\alpha}$.
Equations \eqref{wwa} and \eqref{wwb}, respectively, become
$\tanh \beta=-f(\beta)$ and $\tanh \beta=-1/f(\beta)$;
they clearly do not have positive solutions proving item {\em (i)} in the theorem.

Next, let us assume that $\lambda<-\sqrt{\alpha}$ so that $\lambda=-(\sqrt{\beta^2+\alpha\ell^2})/\ell<-\sqrt{\alpha}$. Equations \eqref{wwa} and \eqref{wwb}, respectively, become
$\tanh \beta=f(\beta)$ and $\tanh \beta=1/f(\beta)$.
Since $\tanh \beta<1$ and $f(\beta)>1$, the first equation has no positive solutions. To study the second equation, we make the change of variables  $\beta=\ell\sqrt{\alpha}\sinh{t}$, $t\geq 0$. Since $f(\ell\sqrt{\alpha}\sinh{t})=\coth{t}$ and $\tanh$ is increasing on $\RR$, thus one-to-one, we obtain the equation $ \tanh \beta=1/f(\beta)$ which is equivalent to $F(t)=0$, where		
the function $F(t):=\ell\sqrt{\alpha}\sinh{t}-t$, $t\geq 0$, has derivative
$F'(t)=\ell\sqrt{\alpha}\cosh{t}-1$. First, if $\ell\sqrt{\alpha}>1$ then $F$ is increasing on $[0,\infty)$ and thus has no positive real roots. If $\ell\sqrt{\alpha}<1$ then $F'$ has a unique root at $t_{\ell,\alpha}^*=\mathrm{arccosh}(\frac{1}{\ell\sqrt{\alpha}})$, is decreasing on $[0, t_{\ell,\alpha}^*]$ and increasing on $[ t_{\ell,\alpha}^*,\infty)$. Since $F(0)=0$ and $\lim_{t\to\infty}{F(t)}=\infty$, we infer that $F$ has exactly one positive solution, proving (ii).

We now consider Case B in \eqref{AB}, that is, consider $\lambda\in(-\sqrt{\alpha}, \sqrt{\alpha})$. We set $\gamma:=(\sqrt{\alpha-\lambda^2})\ell$ so that $i\gamma=\ell\sqrt{\lambda^2-\alpha}$, note that $\gamma\in(0,\ell\sqrt{\alpha} )$ and introduce the function $g(\gamma):=(\sqrt{\alpha\ell^2-\gamma^2})/\gamma$. The function  $g(\cdot)$ monotonically decays from $+\infty$ to $0$ when $\gamma$ changes from $0$ to $\ell\sqrt{\alpha} $. 
We recall that $\tanh(iz)=i\tan(z)$.

First, let us assume that $\lambda\in(0,\sqrt{\alpha})$ so that $\lambda=(\sqrt{\alpha\ell^2-\gamma^2})/\ell$. Equations \eqref{wwa} and \eqref{wwb}, respectively, become
$
\tan\gamma=g(\gamma)$ and $\tan\gamma=-1/g(\gamma)$.
Clearly, the second of the two equations  has no positive solutions $\gamma$ while if 
$\ell\sqrt{\alpha} \in\big(0,\pi\big)$ then the graphs of $g(\cdot)$ and $\tan(\cdot)$ have exactly one intersection corresponding to the point $\gamma_0$ in the interval $(0,\min\{\pi/2,\ell\sqrt{\alpha} \})$ thus proving item {\em (iii)} in the theorem. 

To prove item {\em (iv)}, we assume $\ell\sqrt{\alpha} \in\big(k \pi ,(k+1)\pi\big)$ for some $k=0,1,\ldots$ and find $k+1$ intersections of the graphs corresponding to the points $\gamma_j\in\big(j\pi, \pi/2+j\pi\big)$, $j=0,1,\ldots,k-1$, and the last point $\gamma_k$ that belongs to the interval $\big(k\pi,\min\{\ell\sqrt{\alpha} , \pi/2+k\pi\}\big)$.

Next, let us assume that $\lambda\in(-\sqrt{\alpha},0)$ so that $\lambda=-(\sqrt{\alpha\ell^2-\gamma^2})/\ell$. Equations \eqref{wwa} and \eqref{wwb}, respectively, become
$\tan\gamma=-g(\gamma)$ and $\tan\gamma=1/g(\gamma)$.
Clearly, the first equation has no positive solutions $\gamma$. To deal 
with the second equation we introduce the function $G(\gamma)=1/g(\gamma)-\tan\gamma$, $\gamma\in(0,\ell\sqrt{\alpha} )$, so that $G'(\gamma)=\psi(\gamma)\big((\alpha\ell^2-\gamma^2)^{-3/2}\cos^{-2}\gamma\big)$ where we temporarily denote $\psi(\gamma)=\alpha\ell^2\cos^2\gamma-(\alpha\ell^2-\gamma^2)^{3/2}$.
We notice that $G(0)=0$.

If $\ell\sqrt{\alpha} <1$ then $G'(\gamma)>0$
for all $\gamma\in(0,\ell\sqrt{\alpha} )$ because $\psi(0)>0$ and $\psi'(\gamma)>0$, and thus $G(\cdot)$ has no roots just proving item {\em (v)} in the theorem.

If $\ell\sqrt{\alpha} >1$ then  $G'(0)<0$ and $G(\gamma)\to+\infty$ as $\gamma\to\ell\sqrt{\alpha} ^-$, and thus $G(\cdot)$ has at least one root 
that corresponds to a $\lambda\in(-\sqrt{\alpha},0)$ thus proving item {\em (vi)} in the theorem.

To prove {\em (vii)}, we notice that the graph of the monotonically growing from $0$ to $+\infty$ function $1/g(\gamma)$ as $\gamma$ changes from $0$ to $\ell\sqrt{\alpha} $ has $k$ intersections with the graph of $\tan(\cdot)$ corresponding to some points $\gamma_j\in(j\pi,\pi/2+j\pi)$, $j=0,1,\ldots, k-1$. There is no point $\gamma_k$ on this list because for $j=k\ge1$ the conditions
$\gamma_k\in\big(k\pi, \pi/2+k\pi\big)$ and $\ell\sqrt{\alpha} \in \big(k\pi, \pi/2+k\pi\big)$  for  $k\ge1$
contradict the equality $G(\gamma_k)=0$. Indeed, $G(\gamma)=0$ if and only if  the equality $(2-\alpha\ell^2/\gamma^2)+(-1+\alpha\ell^2/\gamma^2)\cos^{-2}\gamma=0$ holds,  
but the inequality $2<\alpha\ell^2/\gamma^2$ is in contradiction with the conditions.
Finally, each of the points $\gamma_j$ just constructed  gives a respective resonance $\lambda_j$ proving  item {\em (vii)} in the theorem.

To begin the proof of {\em (viii)}, we let $\lambda=iw$ for some $w\in\RR\setminus\{0\}$ and use $\sqrt{\lambda^2-\alpha}=i\sqrt{w^2+\alpha}$ and $\sqrt{-\alpha}=i\sqrt{\alpha}$ to re-write equations \eqref{FOUR3} as follows,
\begin{equation*}
\sqrt{\alpha}e^{i\ell\sqrt{w^2+\alpha}}=\pm i(\sqrt{w^2+\alpha}-w),\,
\sqrt{\alpha}e^{i\ell\sqrt{w^2+\alpha}}=\pm (\sqrt{w^2+\alpha}-w).
\end{equation*}
The first equation yields $\cos(\ell\sqrt{w^2+\alpha})=0$ and $\sqrt{\alpha}\sin(\ell\sqrt{w^2+\alpha})=\pm(\sqrt{w^2+\alpha}-w)$ or, equivalently,
$\ell\sqrt{w^2+\alpha}=\pi/2+k\pi$ and $\sqrt{\alpha}(-1)^k=\pm(\sqrt{w^2+\alpha}-w)$ and therefore we conclude that $\sqrt{w^2+\alpha}=w\pm\sqrt{\alpha}(-1)^k$ for some $k\in\ZZ$ which is not possible since $w\neq0$. The argument for the second equation is analogous, and so \eqref{FOUR3} has no pure imaginary solutions thus proving item {\em (viii)} in the theorem.
\end{proof}
Next, we consider the resonances generated by negative eigenvalues $\alpha\in\spec(a)$ that are located on the real or imaginary axis, if any. 
We begin with the following elementary observations.
\begin{remark}\label{r4.2}
\begin{enumerate}
\item[(i)] The equation $\sqrt{\gamma^2+1}-\ln{(1+\sqrt{\gamma^2+1})}+\ln{\gamma}=0$, $\gamma\in\RR$, has a unique solution denoted by $\gamma_*\in(0,1)$;
\item[(ii)]  For any $\gamma\in(0,\gamma_*)$ the equation $\gamma\cosh{t}=t$ has two roots, $t_1(\gamma)\in\big(0,\arcsinh({1}/{\gamma})\big)$ and $t_2(\gamma)\in\big(\arcsinh({1}/{\gamma}),\infty\big)$;
\item[(iii)]  The equation $\gamma_*\cosh{t}=t$ has a unique root $t=\arcsinh({1}/{\gamma_*})$;
\item[(iv)]   If  $\gamma>\gamma_*$ the equation $\gamma\cosh{t}=t$ has no roots.
\hfill$\Diamond$
\end{enumerate}
\end{remark}
\begin{theorem}\label{t4.3} Assume that $0\notin\spec(a)$ and let $\alpha$ be a negative eigenvalue of $a$. Let $\gamma_*$ be constant introduced in Remark~\ref{r4.2}. Then,
\begin{enumerate}
\item[(i)] If $\ell\sqrt{-\alpha}\in(0,\gamma_*)$ then there are two real $\alpha$-resonances, and both are negative;
\item[(ii)] If $\ell\sqrt{-\alpha}=\gamma_*$ then $\lambda=-{1}/{\ell}$ is the unique real $\alpha$-resonance;
\item[(iii)] If $\ell\sqrt{-\alpha}>\gamma_*$ there are no real $\alpha$-resonances;
\item[(iv)] There are no pure imaginary nonzero $\alpha$-resonances.
\end{enumerate}
\end{theorem}
\begin{proof} We recall that $\lambda$ is $\alpha$-resonance provided $\lambda=\sqrt{\alpha}\cosh{z}$ and $z\in\CC$ satisfies equation \eqref{S0mu}. Imposing the condition $\lambda\in\RR$, denoting $t=\re  z$ and $s=\im z$, and using $\sqrt{\alpha}=\rmi\sqrt{-\alpha}$, we re-write 
\eqref{S0mu} as
\begin{equation}\label{t4.3.1}
\begin{split}
&\lambda=-\sqrt{-\alpha}\sinh{t}\sin{s},\quad
\sqrt{-\alpha}\cosh{t}\cos{s}=0,\\
&-\ell\sqrt{-\alpha}\cosh{t}\sin{s}+t=0,\quad
\ell\sqrt{-\alpha}\sinh{t}\cos{s}+s=\mu.\end{split}
\end{equation}
By the second equation in \eqref{t4.3.1} $s=(k+1/2)\pi$ for some 
$k\in\ZZ$ and so \eqref{t4.3.1} reads
\begin{equation}\label{t4.3.2}
\lambda=-\sqrt{-\alpha}\sinh{\big((-1)^kt\big)},\;s=\mu=(k+1/2)\pi,\;\ell\sqrt{-\alpha}\cosh{\big((-1)^kt\big)}=(-1)^kt.
\end{equation}
If $\ell\sqrt{-\alpha}\in(0,\gamma_*)$ then $(-1)^kt\in\{t_1(\gamma),t_2(\gamma)\}$ by Remark~\ref{r4.2}(ii) and thus the only real $\alpha$-resonances are $-\sqrt{-\alpha}\sinh{\big(t_1(\gamma)\big)}$ and $-\sqrt{-\alpha}\sinh{\big(t_2(\gamma)\big)}$, proving (i). If $\ell\sqrt{-\alpha}=\gamma_*$ then by Remark~\ref{r4.2}(iii) the last equation in \eqref{t4.3.2} has a unique solution $(-1)^{k}t=\arcsinh({1}/{\gamma_*})$ and so 
$\lambda=-\sqrt{\alpha}\sinh{\big((-1)^kt\big)}=-{\sqrt{-\alpha}}/{\gamma_*}=-{1}/{\ell}$, thus proving (ii). Assertion (iii) follows immediately since by Remark~\ref{r4.2}(iv) the last equation in \eqref{t4.3.2} has no solutions in the case when $\ell\sqrt{-\alpha}>\gamma_*$. 
 
To begin the proof of {\em (iv)}, let $\lambda=iw$ for a $w\in\RR\setminus\{0\}$ be a solution to one of the equations in \eqref{FOUR3}. Clearly, $|w|\neq\sqrt{-\alpha}$ as otherwise $\lambda^2=\alpha\in\spec(a)$ and by Lemma \ref{l1l} $\lambda=-1/\ell$ must be real, which is not the case. We will thus consider two cases:
 Case A: $|w|>\sqrt{-\alpha}$ and Case B: $|w|<\sqrt{-\alpha}$.
In Case A $\sqrt{\lambda^2-\alpha}=i\sqrt{w^2+\alpha}$ and thus \eqref{FOUR3} become
\begin{equation*}
\sqrt{-\alpha}e^{i\ell\sqrt{w^2+\alpha}}=\pm(\sqrt{w^2+\alpha}-w),\,
\sqrt{-\alpha}e^{i\ell\sqrt{w^2+\alpha}}=\pm i (\sqrt{w^2+\alpha}-w).
\end{equation*}
The first equation yields $\sin(\ell\sqrt{w^2+\alpha})=0$ and $\sqrt{-\alpha}\cos(\ell\sqrt{w^2+\alpha})=\pm(\sqrt{w^2+\alpha}-w)$, or equivalently, $\ell\sqrt{w^2+\alpha}=k\pi$ and $\sqrt{-\alpha}(-1)^k=\pm(\sqrt{w^2+\alpha}-w)$ for some $k\in\ZZ$, and therefore $\sqrt{w^2+\alpha}=w\pm(-1)^k\sqrt{-\alpha}$ which is not possible as $|w|\neq\sqrt{-\alpha}$. The argument for the second equation is similar completing the proof in Case A.
  In Case B one has $\sqrt{\lambda^2-\alpha}=\sqrt{-w^2-\alpha}\in\RR$ and thus \eqref{FOUR3} becomes
\[
i\sqrt{-\alpha}e^{\ell\sqrt{-w^2-\alpha}}=\pm(\sqrt{-w^2-\alpha}-iw),\,
\sqrt{-\alpha}e^{\ell\sqrt{-w^2-\alpha}}=\pm (\sqrt{-w^2-\alpha}-iw).
\]
The first equation yield $\sqrt{-w^2-\alpha}=0$ which is not possible as $|w|\neq\sqrt{-\alpha}$ while the second equation yield $w=0$ which is not the case either thus completing the proof of item {\em (iv)} in the theorem in Case B.
\end{proof}

\section{Multiplicity of Resonances and the Weyl Law}\label{Multiplicity}
In this section we study multiplicities of resonances of the Schr\"odinger operator $A$. In the sequel we denote by $R_A^\infty(\cdot)$
the finitely meromorphic extension of the resolvent $\lambda\mapsto(\lambda^2-A)^{-1}$  to the entire complex plane $\CC$ so that  $R_A^\infty(\lambda) \in\cC\cL\big(L^2_{\textrm{comp}}(\RR; \CC^d),H^2_{\loc}(\RR; \CC^d)\big)$ for any $\lambda\in\CC$ that is not a zero of the Jost function $\cF$, c.f. \eqref{resA}.
We recall our notation $\oor(\lambda_0,f)$ for the multiplicity of the zero $\lambda_0$ of the holomorphic function $f$ defined in a neighborhood of $\lambda_0$ and $m_\rma(\alpha,a)$ for the algebraic multiplicity of the eigenvalue $\alpha$ of the matrix $a$. 

In the case of matrix valued potentials we choose to use the definition of multiplicity of resonances given in  \cite{G,GH}, cf.\ \eqref{def-multiplicity} below. Our goal is to
reduce the computation of multiplicity of resonances of the linear operator $A$ to the computation of multiplicity of zeros of the function ${\mathcal{F}}(\cdot)={\det}_{d\times d}W(\cdot)$ introduced in Section~\ref{Jost-Green}. To establish this result we need to introduce some additional notation in order to prove a special case of the Birman-Schwinger principle.
We fix $\psi\in L^\infty_{\mathrm{comp}}(\RR,\RR)$ such that $\psi(x)=1$ for all $x\in[-\ell,\ell]$ and $\psi(x)=0$ whenever $|x|\geq \ell+1$. Moreover, for the $d\times d$ matrix $a$, we recall the decomposition
$a=u_a|a|$, where $u_a$ is an unitary matrix and $|a|=(a^*a)^{1/2}$.
We denote by $V_\mathrm{left}, V_{\mathrm{right}}:\RR\to\CC^{d\times d}$ the matrix-valued potentials defined by 
$V_\mathrm{left}=u_a|a|^{1/2}\chi_{[-\ell,\ell]}$ and $V_\mathrm{right}=|a|^{1/2}\chi_{[-\ell,\ell]}$	
so  that $V=V_{\mathrm{left}}V_{\mathrm{right}}$.
Throughout, we denote by $M_W$ the operator acting in $L^2(\RR,\CC^d)$ of multiplication by the matrix-valued function $W(\cdot)$. Similarly, we denote by $M_\psi=M_{\psi I_{d\times d}}$. For the free Laplacian, $\spec(\partial_x^2)=\spec_{\mathrm{ess}}(\partial_x^2)=(-\infty,0]$ and
\begin{equation}\label{free resolvent}
\Big((\lambda^2-\partial_x^2)^{-1}f\Big)(x)=\frac{1}{2\lambda}\int_{\RR}e^{-\lambda|x-y|f(y)\rmd y}\;\mbox{for any}\; f\in L^2(\RR,\CC^d),\,x\in\RR,\,\lambda\in\CC_+.
\end{equation}
The resolvent $\lambda\mapsto(\lambda^2I_{L^2}-\partial_x^2)^{-1}$ can be extended holomorphically  to $\CC\setminus\{0\}$ as a function taking values in $\cC\cL\big(L^2_{\textrm{comp}}(\RR; \CC^d), H^2_{\loc}(\RR; \CC^d)\big)$. We denote this extension by $R_0^\infty(\cdot)$ and note that formula \eqref{free resolvent} holds whenever $\re\lambda<0$, $f\in L^2_{\mathrm{comp}}(\RR,\CC^d)$ if the inverse in the left hand side is replaced by $R_0^\infty(\lambda)$.  Since $V\in L^\infty_{\mathrm{comp}}(\RR,\CC^{d\times d})$ and $\psi\in L^\infty_{\mathrm{comp}}(\RR,\RR)$ we have $M_\psi\in\mathcal{B}(L^2(\RR,\CC^d),L^2_{\mathrm{comp}}(\RR,\CC^d))$ and $M_V\in\cC\cL(L^2_{\textrm{loc}}(\RR; \CC^d),L^2(\RR; \CC^d))$, which implies that $M_VR_0^\infty(\lambda)M_\psi\in\mathcal{B}(L^2(\RR,\CC^d))$ for any $\lambda\in\CC\setminus\{0\}$.
\begin{lemma}\label{l5.1} Assume that $0\notin\spec(a)$. Then 
$M_VR_0^\infty(\lambda)M_\psi$ and $M_{V_\rmr}R_0^\infty(\lambda)M_{V_\rml}$ are trace class operators in $L^2(\RR;\CC^d)$ for any $\lambda\in\CC\setminus\{0\}$.
\end{lemma}	
\begin{proof}
First, we claim that the integral operator
$M_\psi R_0^\infty(\lambda)M_\psi$ with semi-separable kernel is trace class for any $\lambda\in\CC\setminus\{0\}$.
Indeed, by \cite[Theorem 4.5]{GM}  we know that 
$M_\psi R_0^\infty(\lambda)M_\psi$ is a trace class operator whenever $\re\lambda>0$. The same conclusion holds by 
 \cite[Proposition 5.6]{S1}  for any $\lambda\in\rmi\RR\setminus\{0\}$. 
If $\re\lambda<0$ then a simple computation shows that the integral kernel of the operator $M_{\psi}\big(R_0^\infty(\lambda)-R_0^\infty(-\lambda)\big)M_\psi$ is given by the formula
\[
K_{\psi}(x,y)=\frac{1}{2\lambda}\psi(x)\cosh{\big(\lambda(x-y)\big)}\psi(y),\;x,y\in\RR,	
\]
which can be written as the difference of two separable kernels given by functions with compact supports. Hence, the operator of convolution with $K_{\psi}$ in $L^2(\RR,\CC^d)$ is trace class and so is $M_\psi R_0^\infty(\lambda)M_\psi$  whenever $\re\lambda<0$ as claimed. 
Next, since $\psi V=V$, $\psi V_\rml=V_\rml$ and $\psi V_\rmr=V_\rmr$ it follows that 
\begin{equation*}
M_VR_0^\infty(\lambda)M_\psi=M_V\Big(M_\psi R_0^\infty(\lambda)M_\psi\Big),\;M_{V_\rmr}R_0^\infty(\lambda)M_{V_\rml}=M_{V_\rmr}\Big(M_\psi R_0^\infty(\lambda)M_\psi\Big)M_{V_\rml}	
\end{equation*}
are trace class for any $\lambda\in\CC\setminus\{0\}$, proving the lemma.
\end{proof}
Our next result describes the resonances $\lambda\in\cR(A)$ of the Schr\"odinger operator $A$ using the trace class Birman-Schwinger type operators   $M_{V_\rmr}R_0^\infty(\lambda)M_{V_\rml}$ and $M_VR_0^\infty(\lambda)M_\psi$. It involves the methods and results from \cite{GM} on the theory of integral operators with semi-separable kernels, see also \cite[Chapter IX]{GGK1}. Alternatively, one could re-write the eigenvalue problem for the second order Schr\"odinger operator as a first order system of differential equations and apply results from \cite{GLM} relating the Evans (Jost) functions for such systems to Fredholm determinants; we do not pursue this alternative approach here.
\begin{theorem}\label{t5.2}
Assume that $0\notin\spec(a)$. Then
the Jost function $\cF(\lambda)=\det_{\CC^{d\times d}}(W(\lambda))$ and the Fredholm determinants of the operators in $L^2(\RR;\CC^d)$ satisfy the identity
\begin{equation}\label{det-equal}
(2\lambda)^{-d}\cF(\lambda)=\det(I_{L^2(\RR;\CC^d)}-M_VR_0^\infty(\lambda)M_\psi)=\det(I_{L^2(\RR;\CC^d)}-M_{V_\rmr}R_0^\infty(\lambda)M_{V_\rml})
\end{equation}
for any $\lambda\in\CC\setminus\{0\}$. In particular,  
\begin{equation}\label{1-K}
\lambda\in\cR(A)\setminus\{0\}\;\mbox{iff}\;\;1\in\spec_{\mathrm{d}}\big(M_{V_\rmr}R_0^\infty(\lambda)M_{V_\rml}\big)\;\mbox{iff}\;\;1\in\spec_{\mathrm{d}}\big(M_VR_0^\infty(\lambda)M_\psi\big).		
\end{equation}	
\end{theorem}	
\begin{proof}
Since $V=V_\rml V_\rmr$  and $\psi V_\rml=V_\rml$, by \cite[Proposition 0.2.2]{Y} we obtain 
\begin{align}\label{l5.1.3}
\det(I&-M_VR_0^\infty(\lambda)M_\psi)=\det(I-M_{V_\rml}M_{V_\rmr}R_0^\infty(\lambda)M_\psi)=\det(I-M_{V_\rmr}R_0^\infty(\lambda)M_\psi M_{V_\rml})\nonumber\\
&=\det(I-M_{V_\rmr}R_0^\infty(\lambda)M_{V_\rml})\;\text{ for any $\lambda\in\CC\setminus\{0\}$ and $I=I_{L^2(\RR;\CC^d)}$},
\end{align}	
proving the second equality in \eqref{det-equal}. 

To finish the proof of the theorem, we will apply \cite[Theorem 3.2, formula (3.9)]{GM}.  To this end, following \cite{GM}, we first introduce some operators needed to reduce the computation of  the infinite dimensional Fredholm determinants in \eqref{det-equal} to the computation of finite dimensional determinants. Fix $\lambda\in\CC_+$. The trace class operator $M_VR_0^\infty(\lambda)M_\psi$ is an operator with a semi-separable kernel $K_V:\RR^2\to\CC^{d\times d}$ given by the formula
\[
K_V(x,y)=\Phi_1(x)\Psi_1(y),  \text{ for $x>y$ and $K_V(x,y)= \Phi_2(x)\Psi_2(y)$ for $x<y$},\]
where $\Phi_1,\Phi_2,\Psi_1,\Psi_2:\RR\to\CC^{d\times d}$ are defined as follows,
\begin{align}\label{def-Phi-4}
\Phi_1(x)&=e^{-\lambda x}V(x),\quad \Psi_1(x)=(2\lambda)^{-1}e^{\lambda x}\psi(x)I_{d\times d},\nonumber\\ 
\Phi_2(x)&=e^{\lambda x}V(x),\quad\Psi_2(x)=(2\lambda)^{-1}e^{-\lambda x}\psi(x)I_{d\times d}.	
\end{align}	
In addition, we introduce the Volterra operator $\cH:L^2(\RR,\CC^2)\to L^2(\RR,\CC^d)$ defined by 
\begin{equation}\label{def-cH}
(\cH f)(x)=\int_x^\infty \Lambda(x,y)f(y)\rmd y,\,f\in L^2(\RR,\CC^2), x\in\RR,
\end{equation}	 
where $\Lambda:\RR^2\to\CC^{d\times d}$ is given by 
\begin{equation}\label{def-Lambda}
\Lambda(x,y)=\Phi_1(x)\Psi_1(y)-\Phi_2(x)\Psi_2(y)=-(\lambda)^{-1}\sinh{(\lambda(x-y))}\psi(y)V(x).
\end{equation}
From \cite[Theorem 2.4]{GM} we know that $1\notin\spec(\cH)$ which implies that the Volterra equation 
\begin{equation}\label{Volterra-2}
\hPhi_1(x)=\Phi_1(x)+\lambda^{-1}\int_{x}^{\infty}\sinh{(\lambda(x-y))}\psi(y)V(x)\hPhi_1(y,\lambda)\rmd y,\, x\in\RR,	
\end{equation}	
has a unique solution $\hPhi_1=(I_{L^2}-\cH)^{-1}\Phi_1$. Multiplying \eqref{Volterra-Jost} by $V(x)$ and using $\psi V=V$, we infer  $\hPhi_1(\cdot)=V(\cdot)U_+(\cdot,\lambda)$. Formula (3.9) from \cite[Theorem 3.2]{GM} and \eqref{W-Jost-connection} yield
\begin{align}\label{l5.1.final}
\det(I_{L^2(\RR;\CC)}&-M_VR_0^\infty(\lambda)M_\psi)=\det \Big(I_{d\times d}-\int_{-\infty}^\infty\hPhi_1(x)\Psi_1(x)\rmd x\Big)\nonumber\\
&=\det \Big(I_{d\times d}-(2\lambda)^{-1}\int_{-\infty}^\infty e^{\lambda x}V(x)U_+(x,\lambda)\rmd x\Big)=(-2\lambda)^{-d}\cF(\lambda)
\end{align}
for any $\lambda\in\CC_+$. Since $\cF(\cdot)$ is an an entire function and $M_VR_0^\infty(\lambda)M_\psi$ is a holomorphic function on $\CC\setminus\{0\}$, using the uniqueness of holomorphic extensions, we conclude that \eqref{det-equal} holds for any $\lambda\in\CC\setminus\{0\}$, proving the theorem.
\end{proof}
The criterion \eqref{1-K} allows us to define the multiplicity of a nonzero resonance as it is done in \cite{G,GH}. Indeed, we define the multiplicity of the resonance by letting
\begin{equation}\label{def-multiplicity}
m_\rmR(\lambda_0,A):=m_\rma\big(1,M_{V_\rmr}R_0^\infty(\lambda_0)M_{V_\rml}\big)\;\mbox{for any}\;\lambda_0\in\cR(A)\setminus\{0\}.	
\end{equation}
Moreover, \eqref{det-equal} immediately implies that
\begin{equation}\label{multiplicity-equality}
m_\rmR(\lambda_0,A)=m_\rma\big(1,M_VR_0^\infty(\lambda_0)M_\psi\big)=\oor(\lambda_0,\cF)\;\mbox{for any}\;\lambda_0\in\cR(A)\setminus\{0\}.
\end{equation}
Since the function ${\mathcal{F}}(\cdot)={\det}_{d\times d}W(\cdot)$ is entire and we know that $0\in\cR(A)$ if and only if $\cF(0)=0$, we define $m_R(0,A):=\oor(0,\cF)$ in the case when $0\in\cR(A)$. 

Next, we show that our definition of multiplicity is consistent with the definition of multiplicity  given in \cite[Chapter 2]{DZ} in the case of scalar potentials. Assume $d=1$, $\lambda_0\in\cR(A)$ and $\eps>0$ is such that $D(\lambda_0,\eps)\cap\cR(A)=\{\lambda_0\}$. Let  $\cD:\CC\setminus\{0\}\to\CC$ be the function defined by $\cD(\lambda)=\det(I_{L^2(\RR;\CC)}-M_VR_0^\infty(\lambda)M_\psi)$. Combining \eqref{det-equal}, \eqref{multiplicity-equality} with \cite[Theorem 2.8]{DZ} we conclude that 
\begin{equation}\label{mult-resonance}
\frac{1}{2\pi\rmi}\int_{\partial D(\lambda_0,\eps)}\frac{\cD'(\lambda)}{\cD(\lambda)}\rmd\lambda=m_\rmR(\lambda_0,A)=\oor(\lambda_0,\cF)=\mathrm{Rank}\Big(\int_{\partial D(\lambda_0,\eps)} R_A^\infty(\lambda)\rmd\lambda\Big)
\end{equation} 
for any $\lambda_0\in\cR(A)\setminus\{0\}$, showing that the two definitions are consistent for the case of non-zero resonances. Next, assume that $0\in\cR(A)$.  Since $d=1$,  Lemma~\ref{lzero} yields $V=\alpha\chi_{[-\ell,\ell]}$, where $\alpha=(\pi n/\ell)^2\in\spec(a)$ for some $n\in\ZZ\setminus\{0\}$, or
$\alpha=(\pi/(2\ell)+\pi n/\ell)^2\in\spec(a)$ for some $n\in\ZZ$. Since the potential is real, Lemma~\ref{lzero} guaranties that $\oor(0,\cF)=1$, hence $m_\rmR(0,A)=1$,  which is consistent with \cite[Theorem 2.7]{DZ}. 

An interesting question is if the definitions of algebraic multiplicities of resonances given in \cite{G,GH} and \cite{DZ} coincide also in the case of general matrix valued potentials, but we do not pursue this here.
 
So far we have reduced the computation of the multiplicity of resonances of the Schr\"odinger operator $A$ to the computation of the multiplicity of  zeros of the Jost function. Next, we will look at a factorization of the Jost function, see Lemma \ref{l5.2} below, that will allow us to obtain an effective formula for the multiplicities of resonances of $A$, see Theorem \ref{t5.10}.

 We begin by recalling the definition of the Schr\"odinger operator $A_\alpha=\d_{xx}+\alpha\chi_{[-\ell,\ell]}$ acting in the space
$L^2(\RR; \CC)$ of scalar valued function with $\alpha\in\spec(a)$. Also, we recall that for any $\alpha\in\spec(a)$  the set $\cR(A_\alpha)$ of resonances of the operator $A_\alpha$ is equal to the set
$\cR_\alpha$ of $\alpha$-resonances of the operator $A$.
As in Section~\ref{Jost-Green}, see there \eqref{Wcomp} and \eqref{WprodS}, we construct an entire function $W_\alpha$ whose set of zeros is equal to $\cR(A_\alpha)$. As in \eqref{WprodS} we observe that 
\begin{align}\label{W-alpha}
W_\alpha(\lambda)&:=(\sqrt{\lambda^2-\alpha})^{-1}e^{-2\lambda\ell}\big(2\lambda (\sqrt{\lambda^2-\alpha})\cosh(2\ell\sqrt{\lambda^2-\alpha})+(2\lambda^2-\alpha)\sinh(2\ell\sqrt{\lambda^2-\alpha})\big)\nonumber\\
&=2e^{-2\lambda\ell}\big(\lambda\cosh(\ell \sqrt{\lambda^2-\alpha})+\sqrt{\lambda^2-\alpha}\sinh(\ell \sqrt{\lambda^2-\alpha})\big)\nonumber\\
&\quad\times\big(\cosh(\ell \sqrt{\lambda^2-\alpha})+\lambda (\sqrt{\lambda^2-\alpha})^{-1}\sinh(\ell \sqrt{\lambda^2-\alpha})\big)\;\mbox{for}\;\lambda^2\ne\alpha.
\end{align}
We recall the following identity for the algebraic multiplicity of matrices obtained by functional calculus:
If $h$ is an entire function then $\spec(h(a))=\{h(\alpha):\alpha\in\spec(a)\}$ and, for any $\alpha\in\spec(a)$,
\begin{equation}\label{alg-h-a}
m_\rma(h(\alpha), h(a))=\sum_{\beta\in H(\alpha)}m_\rma(\beta,a),\;\mbox{ where }\; H(\alpha):=\{\beta\in\spec(a):h(\beta)=h(\alpha)\}.
\end{equation}
\begin{lemma}\label{l5.2}
Assume that $0\notin\spec(a)$. Then 
\begin{equation}\label{decomp-cF}
\cF(\lambda)=\prod_{\alpha\in\spec(a)}\big(W_\alpha(\lambda)\big)^{m_\rma(\alpha,a)}\;\mbox{for any}\;\lambda\in\CC.	
\end{equation}
In particular, $\cR(A)=\cup_{\alpha\in\spec(a)}\cR_\alpha$.	
\end{lemma}
\begin{proof}
Since $W(\lambda)=h(b(\lambda))$ is given by \eqref{WprodS} as a function of the matrix $b(\lambda)=\sqrt{\lambda^2-a}$ while $W_\alpha(\lambda)=h(\sqrt{\lambda^2-\alpha})$ is given by \eqref{W-alpha}, by the spectral mapping theorem
\begin{equation}\label{l5.2.1}
\spec\big(W(\lambda)\big)=\{W_\alpha(\lambda):\alpha\in\spec(a)\}\;\mbox{for any}\;\lambda\in\CC.
\end{equation}
We now prove the following claim: For any $\alpha_1,\alpha_2\in\spec(a)$, $\alpha_1\ne\alpha_2$ there exists a positive number $\tau(\alpha_1,\alpha_2)>0$ such that
$W_{\alpha_1}(\tau)\ne W_ {\alpha_2}(\tau)$ for any $\tau>\tau(\alpha_1,\alpha_2)$.	
Indeed, seeking a contradiction we suppose there exists a sequence
$\tau_n\to\infty$ so that $W_{\alpha_1}(\tau_n)=W_{\alpha_2}(\tau_n)$, $n\in\NN$. 
  Collecting the growing and decaying exponential terms in the last equality yields  
\begin{align}\label{l5.2.3}
\big(2\tau_n&+2\sqrt{\tau_n^2-\alpha_1}+\alpha_1(\sqrt{\tau_n^2-\alpha_1})^{-1}\big)\big(e^{2\ell(\sqrt{\tau_n^2-\alpha_1}-\sqrt{\tau_n^2-\alpha_2})}-1\big)\nonumber\\
&=2\sqrt{\tau_n^2-\alpha_2}-2\sqrt{\tau_n^2-\alpha_1}+\mathcal{O}(\tau_n e^{-2\ell\tau_n}) \text{ as $n\to\infty$}.	
\end{align}	
Passing to the limit as $n\to\infty$ yields  $4\ell(\alpha_2-\alpha_2)=0$, a contradiction proving  the claim. 

We may now define
$\tau_a=\max\{\tau(\alpha_1,\alpha_2):\alpha_1,\alpha_2\in\spec(a),\alpha_1\ne\alpha_2\}<\infty$ so  that 
\begin{equation}\label{l5.2.5}
W_{\alpha_1}(\tau)\ne W_ {\alpha_2}(\tau)\;\mbox{for any}\;\tau>\tau_a,\,\alpha_1,\alpha_2\in\spec(a),\,\alpha_1\ne\alpha_2.	
\end{equation}	
From \eqref{alg-h-a}, \eqref{l5.2.1} and \eqref{l5.2.5} we conclude that $m_\rma(W_\alpha(\tau), W(\tau))=m_\rma(\alpha,a)$ for any $\alpha\in\spec(a)$ and $\tau>\tau_a$. Since the determinant of a matrix is the product of all of its eigenvalues, counting algebraic multiplicity, we infer that \eqref{decomp-cF} holds true for any $\lambda\in(\tau_a,\infty)$. Since $\cF(\cdot)$ and $W_\alpha(\cdot)$ are entire functions, we conclude that 
\eqref{decomp-cF} holds true for any $\lambda\in\CC$.
\end{proof}
Our next task is to examine the sets $\{\cR_\alpha\}_{\alpha\in\spec(a)}$.  The following example shows that in general they are not mutually disjoint.
\begin{example}\label{e5.3}
We take the matrix $a=\mathrm{diag}(\alpha_1,\alpha_2)$ with $\alpha_1,\alpha_2$ to be chosen later. Let $h_1,h_2:[0,\pi/2)\to\RR$ be the continuous functions defined by 
$h_1(t)=\ell\tan{t}+t-\pi/2$ and $h_2(t)=\ell\tan{t}+t-\pi$.
Since $h_1(0)=-\pi/2$ and $\lim_{t\to\pi/2-0}h_1(t)=\infty$, there exists $t_1\in(0,\pi/2)$ such that $h_1(t_1)=0$. Thus, $h_2(t_1)=-\pi/2$. Since $\lim_{t\to\pi/2-0}h_2(t)=\infty$, there exists $t_2\in(t_1,\pi/2)$ such that $h_2(t_2)=0$. We define $\alpha_j=\sec^2{(t_j)}$,  $z_j=\rmi t_j$ for $j=1,2$. Since $0<t_1<t_2<\pi/2$ we have $0<\alpha_1<\alpha_2$. We now compute
\begin{equation}\label{e5.3.2}
\ell\sqrt{\alpha_j}\sinh{(z_j)}=\rmi\big(\ell\tan{(t_j)}+t_j\big)\in\{\rmi\pi/2,\rmi\pi\},\,j=1,2.	
\end{equation}	
The change of variables as in \eqref{res-ab} leads to $\lambda=\sqrt{\alpha_j}\cosh{(z_j)}=\sec{(t_j)}\cos{(t_j)}=1$ for $j=1,2$, and so from \eqref{S0mu} and \eqref{e5.3.2} we conclude that $1\in\cR_{\alpha_1}\cap\cR_{\alpha_2}\neq\emptyset$. 
\end{example}	
Although  $\cR_{\alpha_1}\cap\cR_{\alpha_2}$ is not necessarily empty, it is always finite, as our next lemma shows.

\begin{lemma}\label{l5.4}
Assume that $0\notin\spec(a)$. Then, $\cR_{\alpha_1}\cap\cR_{\alpha_2}$ is finite for any $\alpha_1,\alpha_2\in\spec(a)$, $\alpha_1\ne \alpha_2$.
\end{lemma}
\begin{proof}
 We will use the following elementary result:
  If $\{\lambda_n\}_{n\in\NN}\subset\CC$ is such that $\lim_{n\to\infty}\re\lambda_n=-\infty$ and $\lim_{n\to\infty}|\im\lambda_n| e^{\ell\re\lambda_n}\in(0,\infty)$, then
\begin{equation}\label{aux-lim-lambda-n}
\lim_{n\to\infty}\frac{\sqrt{\lambda_n^2-\alpha}}{\lambda_n}=\lim_{n\to\infty}\frac{\lambda_n+\sqrt{\lambda_n^2-\alpha_1}}{\lambda_n+\sqrt{\lambda_n^2-\alpha_2}}=1,\;\lim_{n\to\infty}\big(\sqrt{\lambda_n^2-\alpha_1}-\sqrt{\lambda_n^2-\alpha_2}\big)=0.
\end{equation}	
To proceed with the proof of the lemma, seeking a contradiction, we suppose that  
 $\cR_{\alpha_1}\cap\cR_{\alpha_2}$ is infinite for some
 $\alpha_1,\alpha_2\in\spec(a)$, $\alpha_1\ne\alpha_2$. 
 We introduce the sets, cf.\ \eqref{equiv-ab} and \eqref{equiv-cd},
\begin{equation}\label{l5.4.1}
\cR_{\alpha_j}^\pm=\big\{\lambda\in\CC:\pm\alpha e^{2\ell\sqrt{\lambda^2-\alpha}}=(\lambda-\sqrt{\lambda^2-\alpha})^2\big\},\;j=1,2.
\end{equation}
Since $\cR_{\alpha_1}\cap\cR_{\alpha_2}$ is infinite it follows that
either $\cR_{\alpha_1}^\pm\cap\cR_{\alpha_2}^\pm$ is infinite (Case A),  or  $\cR_{\alpha_1}^\pm\cap\cR_{\alpha_2}^\mp$ is infinite (Case B).	
In Case A we have $\{\lambda_n\}_{n\in\NN}\subseteq\cR_{\alpha_1}^\pm\cap\cR_{\alpha_2}^\pm$ such that $\lim_{n\to\infty}\re\lambda_n=-\infty$ and $\lim_{n\to\infty}|\im\lambda_n|=\infty$. Theorem~\ref{t3.14}(ii) now implies $\lim_{n\to\infty}{|\im\lambda_n|}{e^{\ell\re\lambda_n}}={|\sqrt{\alpha_1}|}/{2}$ as needed in \eqref{aux-lim-lambda-n}. Equation \eqref{l5.4.1} implies
\begin{equation}\label{l5.4.3}
e^{2\ell\big(\sqrt{\lambda_n^2-\alpha_1}-\sqrt{\lambda_n^2-\alpha_2}\big)}=(\alpha_1/\alpha_2)(\lambda_n+\sqrt{\lambda_n^2-\alpha_2})^2(\lambda_n+\sqrt{\lambda_n^2-\alpha_1})^{-2},\end{equation} 
and using \eqref{aux-lim-lambda-n} we arrive at $\alpha_1=\alpha_2$, a contradiction.

Similarly, in Case B, take $\{\lambda_n\}_{n\in\NN}\subseteq\cR_{\alpha_1}^\pm\cap\cR_{\alpha_2}^\mp$ such that $\lim_{n\to\infty}\re\lambda_n=-\infty$ and $\lim_{n\to\infty}|\im\lambda_n|=\infty$. Again, Theorem~\ref{t3.14}(ii) implies that $\lim_{n\to\infty}{|\im\lambda_n|}{e^{\ell\re\lambda_n}}={|\sqrt{\alpha_1}|}/{2}$ as needed in \eqref{aux-lim-lambda-n}. Now, in contrast to Case A, $\lambda_n$ satisfies the equation
\begin{equation}\label{l5.4.4}
-e^{2\ell\big(\sqrt{\lambda_n^2-\alpha_1}-\sqrt{\lambda_n^2-\alpha_2}\big)}=(\alpha_1/\alpha_2)(\lambda_n+\sqrt{\lambda_n^2-\alpha_2})^2(\lambda_n+\sqrt{\lambda_n^2-\alpha_1})^{-2}.
\end{equation} 
Using \eqref{aux-lim-lambda-n} gives $\alpha_1=-\alpha_2$. Plugging this back to \eqref{l5.4.4} yields
\begin{equation}\label{l5.4.5}
\frac{e^{2\ell\big(\sqrt{\lambda_n^2+\alpha_2}-\sqrt{\lambda_n^2-\alpha_2}\big)}-1}{\sqrt{\lambda_n^2+\alpha_2}-\sqrt{\lambda_n^2-\alpha_2}}=-\frac{1}{\lambda_n+\sqrt{\lambda_n^2+\alpha_2}}\Big(\frac{\lambda_n+\sqrt{\lambda_n^2-\alpha_2}}{\lambda_n+\sqrt{\lambda_n^2+\alpha_2}}+1\Big),	
\end{equation}	
and passing to the limit as $n\to\infty$ we conclude that $2\ell=0$,  a contradiction.
\end{proof}	
The next main goal in this section is to compute the multiplicity $m_\rmR(\lambda_0,A)$ of a resonance $\lambda_0\in\cR(A)$ via the multiplicities $m_\rma(\alpha,a)$ of the eigenvalues $\alpha\in\spec(a)$. This will be done in two steps: First, in Lemma \ref{l5.5} we will calculate
$m_\rmR(\lambda_0,A)$ using $\oor(\lambda_0,W_\alpha)$ for $W_\alpha$ from \eqref{W-alpha} and, second, in Lemma \ref{l5.9} we will evaluate $\oor(\lambda_0,W_\alpha)$. The final formula is presented in Theorem \ref{t5.10}. 
To start, for any fixed $\lambda_0\in\cR(A)$ we define a subset of $\spec(a)$,
\begin{equation}\label{def-J}
\cJ(\lambda_0)=\{\alpha\in\spec(a):\lambda_0\in\cR_\alpha\},	
\end{equation}
that might contain more than one element, cf.\ Example~\ref{e5.3}.  	
\begin{lemma}\label{l5.5}
Assume that $0\notin\spec(a)$. Then, 
\begin{equation}\label{mult-formula}
m_\rmR(\lambda_0,A)=\sum_{\alpha\in\cJ(\lambda_0)}m_\rma(\alpha,a)\cdot\oor(\lambda_0,W_\alpha)\;\mbox{for any}\;\lambda_0\in\cR(A).
\end{equation}
\end{lemma}
\begin{proof}
The lemma follows shortly from Lemma~\ref{l5.2}, \eqref{multiplicity-equality} and the definition of $m_\rmR(0,A)$.
\end{proof} 
To determine  the order of each zero of the entire complex valued function $W_\alpha$ we begin by recalling the following well-known result on  zeros of a composition of functions.
\begin{remark}\label{r5.6}
Assume that $g:\Omega_1\to\Omega_2$	and $f:\Omega_2\to\CC$ are holomorphic functions, and that $\lambda_0$ is a zero of order $k_1$ of the function $f(\cdot)$ while $\nu_0$ is a zero of order $k_2$ of $g(\cdot)-\lambda_0$. Then $\nu_0$ is a zero of order $k_1k_2$ of the composition $f\circ g$.
\end{remark}	
To find the order of a zero $\lambda_0$ of the function $W_\alpha$ for some $\alpha\in\spec(a)$, we need to find the order of zeros of two functions: First, of
the function used to perform the change of variables in Section~\ref{Complex-Resonances}, that is, $\sqrt{\alpha}\cosh{(\cdot)}-\lambda_0$ with $\lambda_0$ being a resonance of $A$, and, second, of the function $f_\alpha(\cdot)-\rmi\mu$, with $\mu\in\pi\ZZ\cup(\pi\ZZ+\pi/2)$, where $f_\alpha(z)=\ell\sqrt{\alpha}\sinh z+z$ is the function used  in the proof of Theorem~\ref{t3.1}. In the next remark we collected some easily checkable  facts about the zeros.
\begin{remark}\label{r5.7} 
\begin{enumerate}
\item[(i)] If $\alpha\in\CC\setminus\{0\}$ and $\lambda_0\in\cR_\alpha$ then all zeros of $\sqrt{\alpha}\cosh{(\cdot)}-\lambda_0$ are of order $1$, unless $\alpha=1/\ell^2$ and $\lambda_0=-1/\ell$. In the later case, the set of zeros of $1/\ell\cosh{(\cdot)}+1/\ell$ is $(2\ZZ+1)\pi\rmi$, all of which are of order 2;
\item[(ii)] The function $f_\alpha(\cdot)-\rmi\mu$ has no zeros of order $4$ or higher for any $\alpha\in\CC\setminus\{0\}$, $\mu\in\pi\ZZ\cup(\pi\ZZ+\pi/2)$;
\item[(iii)] The function $f_\alpha(\cdot)-\rmi\mu$ has a zero of order $3$ if and only if $\alpha=1/\ell^2$ and $\mu\in(2\ZZ+1)\pi\rmi$, and in this case  the zero is equal to $\rmi\mu$.\hfill$\Diamond$
\end{enumerate}	
\end{remark}	
The zeros of order $2$ of the function $f_\alpha(\cdot)-\rmi\mu$
will require some work culminating in Lemma \ref{l5.8}. To begin their study, we introduce the following sets, $\Omega^\pm_j$, $j=1,2$, each containing infinitely many elements, and, moreover, containing infinitely many elements in the interval $(1/\ell^2,\infty)$,
\begin{equation}\label{def-Omega}\begin{split}
\Omega_1^\pm&=\{\alpha\in\CC:\ell\sqrt{\alpha}\sinh{(\sqrt{1-\ell^2\alpha})}=\pm\sqrt{1-\ell^2\alpha}, \ell\sqrt{\alpha}\cosh{(\sqrt{1-\ell^2\alpha})}=\pm1\},\\
\Omega_2^\pm&=\{\alpha\in\CC:\ell\sqrt{\alpha}\cosh{(\sqrt{1-\ell^2\alpha})}=\pm\rmi\sqrt{1-\ell^2\alpha}, \ell\sqrt{\alpha}\sinh{(\sqrt{1-\ell^2\alpha})}=\pm\rmi\}\\
\Omega&=\Omega_1^+\cup\Omega_1^-\cup\Omega_2^+\cup\Omega_2^-.\end{split}
\end{equation}	
One can readily check the following assertions to be used later,
\begin{align}\label{prop-Omega-pm}
&\ell\sqrt{\alpha}\in\Omega_j^+\;\mbox{ if and only if }\;-\ell\sqrt{\alpha}\in\Omega_j^-,\;\mbox{for}\; j=1,2,\\
\label{prop-Omega-1-plus}
& \alpha\in\Omega_1^+\;\mbox{ and }\;\sqrt{1-\ell^2\alpha}\in\pi\rmi\ZZ\;\mbox{ implies }\;\alpha=1/\ell^2,\\
\label{prop-Omega-2}
&\sqrt{1-\ell^2\alpha}\notin\pi\rmi\ZZ\;\mbox{ for any }\;\alpha\in\Omega_1^-\cup\Omega_2^+\cup\Omega_2^-.
\end{align}
The zeros of order $2$ of the function $f_\alpha(z)=\ell\sqrt{\alpha}\sinh z+z$ are described next.
\begin{lemma}\label{l5.8}
\begin{enumerate}
\item[(i)] If $\alpha\in\Omega_1^+\setminus\{1/\ell^2\}$ and $k\in 2\ZZ+1$ then, $k\pi\rmi\pm\sqrt{1-\ell^2\alpha}$ are the only zeros of order $2$ of $f_\alpha(\cdot)-k\pi\rmi$;
\item[(ii)] If $\alpha\in\Omega_1^-$ and $k\in 2\ZZ$ then, $k\pi\rmi\pm\sqrt{1-\ell^2\alpha}$ are the only zeros of order $2$ of $f_\alpha(\cdot)-k\pi\rmi$;
\item[(iii)] If $\alpha\in\Omega_2^+$ and $k\in 2\ZZ$ then $k\pi\rmi+\pi\rmi/2+\sqrt{1-\ell^2\alpha}$ is the only zeros of order $2$ of $f_\alpha(\cdot)-k\pi\rmi-\pi\rmi/2$;
\item[(iv)] If $\alpha\in\Omega_2^-$ and $k\in 2\ZZ+1$ then $k\pi\rmi+\pi\rmi/2+\sqrt{1-\ell^2\alpha}$ is the only zeros of order $2$ of $f_\alpha(\cdot)-k\pi\rmi-\pi\rmi/2$;
\item[(v)] The function $f_\alpha(\cdot)-\rmi\mu$ has no other zeros of order $2$ except the ones listed in (i)-(iv).
\item[(vi)] The function $f_{1/\ell^2}(\cdot)-\rmi\mu$ has no zeros of order $2$ for any $\mu\in\pi\ZZ\cup(\pi\ZZ+\pi/2)$. 
\end{enumerate}
\end{lemma}	
\begin{proof}
Fix $\mu\in\pi\ZZ\cup(\pi\ZZ+\pi/2)$. To prove the lemma we solve the system of equations $f_\alpha(z)=\rmi\mu$ and $f_\alpha'(z)=0$, imposing the condition $f_\alpha''(z)\ne 0$. The latter follows from \eqref{prop-Omega-1-plus} and \eqref{prop-Omega-2}. Solving the system we infer that $\rmi\mu+\pm\sqrt{1-\ell^2\alpha}$ are the only two candidates to be zeros of order $2$. Using this information in the equations $f_\alpha(z)=\rmi\mu$ and $f_\alpha'(z)=0$, and the definitions of $\Omega_j^\pm$, $j=1,2$, in \eqref{def-Omega} we immediately conclude assertions \textit{(i)-(v)}. Assertion \textit{(vi)} is straightforward.
\end{proof}
We are ready to calculate the order of zeros of $W_\alpha$ (recall notation in \eqref{def-Omega}).
\begin{lemma}\label{l5.9}
Assume that $0\not\in\spec(a)$ and let $\alpha$ be an eigenvalue of the matrix $a$. Then
\begin{enumerate}
\item[(i)] If $\alpha\in\CC\setminus\Omega$, then  $\oor(\lambda_0,W_\alpha)=1$ for any $\lambda_0\in\cR_\alpha$;
\item[(ii)] If $\alpha\in\Omega\setminus\{1/\ell^2\}$ then $\oor(\lambda_0,W_\alpha)=1$ if $\lambda_0\in\cR_\alpha\setminus\{-1/\ell\}$ and $\oor(-1/\ell,W_\alpha)=2$;
\item[(iii)] If $\alpha=1/\ell^2$ then $\oor(\lambda_0,W_{1/\ell^2})=1$ for any $\lambda_0\in\cR_{1/\ell^2}$. 
\end{enumerate}
\end{lemma}
\begin{proof}
Fix $\alpha\in\spec(a)$ and $\lambda_0\in\cR_\alpha$ and assume $(\alpha,\lambda_0)\ne(1/\ell^2,-1/\ell)$. Then,  $\lambda_0^2\ne\alpha$. Hence, we can find $\eps>0$ such that $\lambda^2\ne\alpha$ for any $\lambda\in D(\lambda_0,\eps)$. Then, there exists $\mu\in\pi\ZZ\cup(\pi\ZZ+\pi/2)$ and a zero $z_0$ of the function $f_\alpha(\cdot)-\rmi\mu$ such that $\sqrt{\alpha}\cosh{z_0}=\lambda_0$. By continuity, there exists $\delta>0$ such that $\sqrt{\alpha}\cosh{z}\in D(\lambda_0,\eps)$ for any $z\in D(z_0,\delta)$.
Making again the change of variables $\lambda=\sqrt{\alpha}\cosh{z}$ with $z\in D(z_0,\delta)$ and using the arguments from Section~\ref{Jost-Green}, we infer that  there exists a holomorphic in $D(z_0,\delta)$ function $\tW_\alpha$ such that $\tW_\alpha(z)\ne 0$ for any $z\in D(z_0,\delta)$ and
\begin{equation}\label{l5.9.1}
W_\alpha(\sqrt{\alpha}\cosh{z})=(e^{2lf_\alpha(z)}-1)(e^{2lf_\alpha(z)}+1)\tW_\alpha(z)\;\mbox{for any}\;z\in D(z_0,\delta).
\end{equation}	
All zeros of the entire functions $\mu\mapsto e^\mu-1$ and $\mu\mapsto e^\mu+1$ are of order 1. Assertions \textit{(i)} and \textit{(ii)} of the lemma  follow  from Remark~\ref{r5.6}, Remark~\ref{r5.7}, Lemma~\ref{l5.8} and \eqref{l5.9.1}. Similarly, we see that $\oor(\lambda_0,W_{1/\ell^2})=1$ for any $\lambda_0\in\cR_{1/\ell^2}\setminus\{-1/\ell\}$. In the special case $\lambda_0=-1/\ell$, the conclusion of assertion \textit{(iii)} follows from Lemma~\ref{l1l}.
\end{proof}

Summarizing in the next theorem the results obtained  in Lemma~\ref{l5.5} and Lemma~\ref{l5.9} we arrive at the final formula for  the multiplicity of a resonance of the Schr\"odinger operator $A$. We define the function $\mathfrak{m}:\CC\times\CC\setminus\{0\}\to\NN$ by letting
$\mathfrak{m}(\lambda,\alpha)=2$  if $\lambda=-1/\ell$ and $\alpha\in\Omega\setminus\{1/\ell^2\}$ and $\mathfrak{m}(\lambda,\alpha)=1$ otherwise and recall that the set $\cJ(\lambda_0)$ is defined in \eqref{def-J}.
\begin{theorem}\label{t5.10}
Assume that $0\notin\spec(a)$. Then for any $\lambda_0\in\cR(A)\setminus\{0\}$ we have  
\begin{equation}\label{mult-formula-2}
m_\rmR(\lambda_0,A)=\sum_{\alpha\in\cJ(\lambda_0)}m_\rma(\alpha,a)\mathfrak{m}(\lambda_0,\alpha).
\end{equation}
\end{theorem}	

In the remaining part of this section we discuss the Weyl Law for the number of resonances of the Schr\"odinger operator $A$ with the matrix valued square well potential, that is, we determine the asymptotic behavior of the number of resonances, counting multiplicities, located in a large disk centered at $0$. This type of problems has a long history, and we refer to \cite{CH,DZ,Sj2,Z,Z2} for its extensive review and discussions of its importance in various settings. Our objective here is to merely demonstrate how the matrix valued case could be reduced to the well studied scalar valued case. We define the resonance counting function for $A$ by
\begin{equation}\label{cal-N-A}
\cN_A(R)=\sum_{\lambda_0\in\cR(A)\cap D(0,R)}	m_\rmR(\lambda_0,A),\;\mbox{for}\; R>0
\end{equation}  
and  the resonance counting function
for $A_\alpha=\partial_{xx}+\alpha\chi_{[-\ell,\ell]}$, $\alpha\in\spec(a)$, 
by 
\begin{equation}\label{cal-N-A-alpha}
\cN_{A_\alpha}(R)=\sum_{\lambda_0\in\cR_\alpha\cap D(0,R)}	m_\rmR(\lambda_0,A_\alpha),\;\mbox{for}\; R>0.
\end{equation}  
\begin{theorem}\label{t5.11}
Assume $0\notin\spec(a)$ and let $A=\partial_{xx}+a\chi_{[-\ell,\ell]}$, where $a\in\CC^{d\times d}$ is a matrix. Then, 
\begin{equation}\label{Weyl-lim }
\lim_{R\to\infty}\big({\cN_A(R)}/{R}\big)={4\ell d}/{\pi}.	
\end{equation}		
\end{theorem}	
\begin{proof}
Applying \eqref{multiplicity-equality} for the case of $d=1$, we infer that $m_\rmR(\lambda_0,A_\alpha)=\oor(\lambda_0,W_\alpha)$ for any $\lambda_0\in\cR(A)\setminus\{0\}$. The equality holds for $\lambda_0=0$ as well by our definition of $m_\rmR(0,A_\alpha)$, provided $0\in\cR_\alpha$. From Lemma~\ref{l5.5} we obtain 
\begin{equation}\label{t5.11.1}
m_\rmR(\lambda_0,A)=\sum_{\alpha\in\cJ(\lambda_0)}m_\rma(\alpha,a)\cdot m_\rmR(\lambda_0,A_\alpha)\;\mbox{for any}\;\lambda_0\in\cR(A),
\end{equation}	
where $\cJ(\lambda_0)$ is defined in \eqref{def-J}.  Since $\cR(A)=\cup_{\alpha\in\spec(a)}\cR_\alpha$, by \eqref{cal-N-A}, \eqref{cal-N-A-alpha} and \eqref{t5.11.1}  
\begin{equation}\label{t5.11.2}
\cN_A(R)=\sum_{\alpha\in\spec(a)}m_\rma(\alpha,a)\cN_{A_\alpha}(R)\;\mbox{for any}\;R>0.	
\end{equation}	
As pointed out in \eqref{mult-resonance}, in the case of scalar potentials the definition of multiplicity of a non-zero resonance given in the current paper, cf.\ \eqref{def-multiplicity}, coincides with that given in \cite[Definition 2.3]{DZ}. Hence, from \cite[Theorem 2.16]{DZ} (see also \cite{Z2}), we obtain 
\begin{equation}\label{t5.11.3}
\lim_{R\to\infty}\big({\cN_{A_\alpha}(R)}/{R}\big)={4\ell}/{\pi}\;\mbox{ for any }\;\alpha\in\spec(a).	
\end{equation}
The conclusion of the theorem now follows from \eqref{t5.11.3} by passing to the limit in \eqref{t5.11.2} since $\sum_{\alpha\in\spec(a)}m_\rma(\alpha)=d$. 	 		
\end{proof}
\appendix

\section{Properties of Solutions of the system $S_0(p,q,\mu)$}\label{App}
In this appendix we discuss several technical lemmas that describe the solutions of the system $(S_0(p,q,\mu))$ introduced in Section~\ref{Complex-Resonances}. In particular, Lemma \ref{l3.2} yields the existence of infinitely many resonances in each quadrant of the left half plane while Lemmas \ref{l3.5} and \ref{l3.8} provide main technical ingredients of the proof of Lemma \ref{l3.10}. 
\begin{lemma}\label{l3.2}
Assume $p\geq 0$ and $(p,q)\ne (0,0)$. Then system $(S_0(p,q,\mu))$ has infinitely many solutions $(t,s)$ in each quadrant of $\RR^2$. 	
\end{lemma}	
\begin{proof}  We re-write the first equation in $(S_0(p,q,\mu))$ as  $R(t)\cos{(s+\varphi(t))}=-t$ and
express $s$ as a function of $t$ by $s=s_k^\pm(t)$ where
we introduced the functions
 $R$, $\varphi$, $s_k^\pm$ such that
\begin{equation*}\label{l3.2.1}\begin{split}
R(t)&=\sqrt{(p\sinh{t})^2+(q\cosh{t})^2},\, \cos{(\varphi(t))}={p\sinh{t}}/{R(t)},\,\sin{(\varphi(t))}={q\cosh{t}}/{R(t)},\\
s_k^\pm(t)&=\pm\arccos(t/R(t))-\varphi(t)+(2k\mp 1)\pi, \, k\in\ZZ,\, |\varphi(t)|\le\pi/2,\,
\, |t|\ge t_*,\end{split}	
\end{equation*}
and choose $t_*$ such that $R(t)\geq |t|$ whenever $|t|\geq t_*>0$. The choice of $t_*$ is possible because  by the assumptions in the lemma
\begin{equation}\label{l3.2.2}
\lim_{t\to\pm\infty}{{R(t)}{e^{-|t|}}}=\sqrt{p^2+q^2}>0,\quad	\lim_{t\to\pm\infty}{{t}/{R(t)}}=0.
\end{equation}
 Plugging $s=s_k^\pm(t)$ into the second equation of \eqref{sys-S0mu},  a  computation shows that $(t,s)$ is a solution to \eqref{sys-S0mu} if and only if for some $k\in\ZZ$ it is a solution to one of the following two systems,
$s=s_k^+(t)$, $F_+(t)+s_k^+(t)=\mu$ or $s=s_k^-(t)$, $F_-(t)+s_k^-(t)=\mu$, where we introduce notation
\[ F_\pm(t)=\big(\mp(p^2+q^2)\sinh{t}\cosh{t}\sqrt{R^2(t)-t^2}+pqt\big)/{R^2(t)}.\]
Clearly, $s_k^+(t)\to\pm\infty$ and $s^-_k(t)\to\pm\infty$ uniformly for $|t|\ge t_*$ as $k\to\pm\infty$ while $F_+(t)\to\mp\infty$ and $F_-(t)\to\pm\infty$ as $t\to\pm\infty$. 
We show that there are infinitely many $(t,s)$ in the first quadrant as follows:  $s_k^+(t)>0$ for $|t|\ge t_*$ and $F_+(t_*)+s_k^+(t_*)>\mu$ 
for each $k>0$ sufficiently large because $s_k^+(t)\to+\infty$ as $k\to\infty$. We can find a $t_k\ge t_*>0$ so that $F_+(t_k)+s_k^+(t_k)=\mu$ because $F_+(t)\to-\infty$ as $t\to+\infty$. Now $(t_k,s_k^+(t_k))$ is the desired solution. For the fourth quadrant we argue as follows:  $s_k^-(t)<0$ for $|t|\ge t_*$ and $F_-(t_*)+s_k^-(t_*)<\mu$ for each $k<0$ with $|k|$ sufficiently large because $s_k^-(t)\to-\infty$ as $k\to-\infty$. We can find a $t_k\ge t_*>0$ so  that $F_-(t_k)+s_k^-(t_k)=\mu$ because $F_-(t)\to\infty$ as $t\to+\infty$. Now $(t_k,s_k^-(t_k))$ is the desired solution. The arguments for the second and third quadrants are analogous.
\end{proof}
\begin{remark}\label{r3.4}
The study of solutions $(t,s)$ to \eqref{sys-S0mu} can be reduced to the case of solutions located in the first quadrant. Indeed,
\begin{enumerate}
\item[(i)] If $(t,s)$ is a solution of $S_0(p,q,\mu)$ then $(-t,s)$ is a solution of $S_0(p,-q,\mu)$; 
\item[(ii)] If $(t,s)$ is a solution of $S_0(p,q,\mu)$ then $(t,-s)$ is a solution of $S_0(p,-q,-\mu)$;
\item[(iii)] If $(t,s)$ is a solution of $S_0(p,q,\mu)$ then $(-t,-s)$ is a solution of $S_0(p,q,-\mu)$.  
\end{enumerate}
\end{remark}

A goal of this appendix is to relate $e^t$ and $s$ as $t\to\infty$ and $s\to\infty$ for $(t,s)$ being  solutions of  system \eqref{sys-S0mu}. We will deal with the equivalent reformulation of  
\eqref{sys-S0mu} given in \eqref{l3.11.1a}, \eqref{l3.11.1b}.
The next lemma
shows how to control the term $q\cos{s}+p\sin{s}$ in the equations. We start with the case when $q\ne0$.
\begin{lemma}\label{l3.5}
Assume $p\geq 0$ and $q\ne 0$, let $\{\mu_n\}_{n\in\NN}$ be a finite sequence and $\{(t_n,s_n)\}_{n\in\NN}$ be such a sequence of solutions  to $(S_0(p,q,\mu_n))$ that $t_n\to\infty$ and $s_n\to\infty$ as $n\to\infty$. Then
\begin{align}\label{u-conv-case 1}
q\cos{(s_n)}+p\sin{(s_n)}&\to-\sqrt{p^2+q^2}\;\mbox{as}\; n\to\infty,\\
\label{t-s-q1-case 1}
\big|s_ne^{-t_n}+\big(q\cos(s_n)+p\sin(s_n)\big)/{2}\big|&\leq \big(\max_{n\in\NN}|\mu_n|+({|p|+|q|})/{2}\big)e^{-t_n}, \, n\in\NN.
\end{align}
\end{lemma}
\begin{proof}
For brevity, we introduce the bounded sequences 
\begin{equation}\label{l3.5.1}\begin{split}
u_n&=q\cos{(s_n)}+p\sin{(s_n)},\;v_n=p\cos{(s_n)}-q\sin{(s_n)},\\
\tu_n&=-q\cos{(s_n)}+p\sin{(s_n)},\;\tv_n=-p\cos{(s_n)}-q\sin{(s_n)}. 	
\end{split}\end{equation}
Since $\{(t_n,s_n)\}_{n\in\NN}$ is a sequence of solutions of $(S_0(p,q,\mu_n))$, 
\begin{equation}\label{l3.5.2}
v_ne^{t_n}+\tv_ne^{-t_n}+2t_n=0,\,
u_ne^{t_n}+\tu_ne^{-t_n}+2s_n=2\mu_n.
\end{equation}
Solving the first equation in \eqref{l3.5.2} for $v_n$ and using that $\{\tv_n\}_{n\in\NN}$ is bounded yields 
\begin{equation}\label{l3.5.4} 
v_n=-\big(2t_n+\tv_ne^{-t_n}\big)e^{-t_n}\to 0\;\mbox{as}\;n\to\infty.
\end{equation}
\noindent{\textbf{Claim 1.}} There exist $c_1,c_2>0$ such that 
$c_1\leq s_ne^{-t_n}\leq c_2$ for any $n\in\NN$; consequently,
\begin{equation}\label{l3.5.8}
\lim_{n\to\infty}{t_n}/{s_n}=0.
\end{equation}
Indeed, by the second equation in \eqref{l3.5.2}
\begin{equation}\label{l3.5.6}
s_ne^{-t_n}=-{u_n}/{2}+\big(2\mu_n-\tu_n\big)e^{-t_n}/{2}\;\mbox{for any}\; n\in\NN.
\end{equation}
Thus the sequence $\{s_ne^{-t_n}\}$ is bounded from above since  $\{u_n\}$, $\{\tu_n\}$ and $\{\mu_n\}$ are bounded and $t_n\to\infty$ as $n\to\infty$. 
To show that $\liminf_{n\to\infty}s_ne^{-t_n}>0$, seeking a contradiction and passing to a subsequence, we may suppose that $s_{n}e^{-t_{n}}\to0$ as $n\to\infty$.
 Then $u_{n}\to0$ by \eqref{l3.5.6}. 
 Using the first line in \eqref{l3.5.1} we arrive at a contradiction,
\begin{equation}\label{l3.5.7}	 	
\cos{(s_{n})}=({qu_{n}+pv_{n}})/({p^2+q^2})\to0\;\mbox{and}\;\sin{(s_{n})}=({pu_{n}-qv_{n}})/({p^2+q^2})\to0
\end{equation}
as $n\to\infty$, thus finishing the proof of Claim 1.   

\noindent{\textbf{Claim 2.}} $s_n\notin\pi\ZZ+{\pi}/{2}$ for all $n\in\NN$ and $\lim_{n\to\infty}\tan{(s_n)}={p}/{q}$.
Indeed, viewing $(S_0(p,q,\mu))$ as a system of equations for  $\sinh{(t_n)}\cos{(s_n)}$ and $\cosh{(t_n)}\sin{(s_n)}$ yields 
\begin{equation}\label{l3.5.10}
\sinh{(t_n)}\cos{(s_n)}=\frac{-pt_n-q(s_n-\mu_n)}{p^2+q^2}\;\mbox{and}\;\cosh{(t_n)}\sin{(s_n)}=\frac{-p(s_n-\mu_n)+qt_n}{p^2+q^2}.
\end{equation}
Since $q\ne0$, by \eqref{l3.5.8} and \eqref{l3.5.10}  $|\sinh{(t_n)}\cos{(s_n)}|\to\infty$ and thus, renumbering $s_n$ if needed, $\cos{(s_n)}\ne0$ as required in Claim 2.
Now \eqref{l3.5.10} yields
\begin{equation}\label{l3.5.12}
\coth{(t_n)}\tan{(s_n)}=\big({p(s_n-\mu_n)-qt_n}\big)/\big({pt_n+q(s_n-\mu_n)}\big)
\end{equation}
and since $t_n\to\infty$  equations \eqref{l3.5.8} and \eqref{l3.5.12} imply $\tan{(s_n)}\to{p}/{q}$ as $n\to\infty$ thus completing the proof of Claim 2.

Next, we denote $m_n=\lfloor {s_n}/{\pi}+{1}/{2}\rfloor$, the integer part, and $\theta_n=\arctan{(\tan{(s_n)})}$, $n\in\NN$. Clearly, from the second assertion in Claim 2 we have
\begin{equation}\label{l3.5.14} 
m_n\to\infty,\;\theta_n\to\arctan{(p/q)}\;\mbox{as}\;n\to\infty.	
\end{equation}
while by the first assertion 
\begin{equation}\label{l3.5.15}
m_n\pi-\frac{\pi}{2}<s_n<m_n\pi+\frac{\pi}{2}\;\mbox{for any}\;n\in\NN,
\end{equation}
which implies that 
\begin{equation}\label{l3.5.16}
s_n=m_n\pi+(s_n-m_n\pi)=m_n\pi+\arctan{(\tan{(s_n-m_n\pi)})}=m_n\pi+\theta_n\;\mbox{for any}\;n\in\NN.
\end{equation}

\noindent{\textbf{Claim 3.}} $(-1)^{m_n}=-\sgn{(q)}$ for all $n\in\NN$ large enough.
Indeed, \eqref{l3.5.16} gives
\begin{equation}\label{l3.5.18}
u_n=\cos{(s_n)}(q+p\tan{(s_n)})=(-1)^{m_n}\cos{(\theta_n)}(q+p\tan{(s_n)}), \, n\in\NN.
\end{equation}
By Claim 2, \eqref{l3.5.14} and \eqref{l3.5.18} we infer 
\begin{equation}\label{l3.5.19}
(-1)^{m_n}u_n\to\sgn{(q)}\sqrt{p^2+q^2}\ne0\;\mbox{as}\;n\to\infty.
\end{equation}
Moreover, from the second equation in \eqref{l3.5.2} we have
\begin{equation}\label{l3.5.20}
u_ne^{t_n}=2(\mu_n-y_n)-\tu_ne^{-t_n}\to-\infty\;\mbox{as}\;n\to\infty.		
\end{equation}
From \eqref{l3.5.19} and \eqref{l3.5.20} we infer 
\begin{equation}\label{l3.5.21}
(-1)^{m_n}e^{t_n}\to-\sgn{(q)}\cdot(+\infty)\;\mbox{as}\;n\to\infty,
\end{equation}
proving Claim 3. 

Going back to \eqref{l3.5.19} we can now conclude that 
\begin{equation}\label{l3.5.22}
u_n=q\cos{(s_n)}+p\sin{(s_n)}\to-\sqrt{p^2+q^2},	
\end{equation}	
thus proving \eqref{u-conv-case 1}. The estimate in \eqref{t-s-q1-case 1} follows by rearranging terms in \eqref{l3.5.6}. 
\end{proof}
Our next lemma shows that the conclusions of Lemma~\ref{l3.5} remain true in the special case when $p>0$ and $q=0$;  it has the same conclusion regarding the asymptotic behavior of ${s_n}{e^{-t_n}}$.
\begin{lemma}\label{l3.8}	
Assume $p>0$ and $q=0$, let $\{\mu_n\}_{n\in\NN}$ be a finite sequence and $\{(t_n,s_n)\}_{n\in\NN}$ be a sequence of solutions to $(S_0(p,0,\mu_n))$
such that $t_n\to\infty$ and $s_n\to\infty$ as $n\to\infty$. Then
\begin{align}\label{u-conv-case 2}
\sin{(s_n)}&\to-1\;\mbox{as}\; n\to\infty,\\
\label{t-s-q1-case 2}
\big|{s_n}{e^{-t_n}}+{p\sin{(s_n)}}/{2}\big|&\leq \big(\max_{n\in\NN}|\mu_n|+{p}/{2}\big)e^{-t_n}\;\mbox{for any}\; n\in\NN.
\end{align}
\end{lemma}
\begin{proof}
Viewing $(S_0(p,0,\mu_n))$ as a system of equations for $\cos{(s_n)}$ and $\sin{(s_n)}$ and using
 $t_n\to\infty$, $s_n\to\infty$ and boundedness of $\{\mu_n\}$ we conclude that 
\begin{equation}
\cos{(s_n)}=-\frac{t_n}{p\sinh{(t_n)}}<0\;\mbox{and}\;\sin{(s_n)}=-\frac{s_n-\mu_n}{p\cosh{(t_n)}}<0 
 \end{equation}
 for all $n\in\NN$ sufficiently large. Solving the first equation yields
\begin{equation}\label{l3.8.1}	 
s_n=\arccos\big(\frac{t_n}{p\sinh{(t_n)}}\big)+(2k+1)\pi, \, k\in\NN,
\end{equation}
and then
\begin{equation}\label{l3.8.2}
\sin{(s_n)}=-\sin{\big(\arccos\big(\frac{t_n}{p\sinh{(t_n)}}\big)\big)}\to-1\;\mbox{as}\;n\to\infty,
\end{equation}	
proving \eqref{u-conv-case 2}. Inequality \eqref{t-s-q1-case 2} follows by rearranging the terms in the second equation of $(S_0(p,0,\mu_n))$.
\end{proof}

\end{document}